%% file: ms.tex
\documentclass{eceasst}
\usepackage{rotating,graphs,graphicx}
\usepackage{amsmath,amssymb}
\usepackage{stmaryrd}
\usepackage{proof}
\usepackage{url}

\usepackage{setspace}

\usepackage{fancybox}
\usepackage{tikz}

\usetikzlibrary{chains,fit,shapes,calc}
\usepackage{enumerate}
\theoremstyle{definition}
\newtheorem{case}{Case}
\newtheorem{subcase}{Case}
\numberwithin{subcase}{case}

\input{frontmatter}

\title{A Unification Algorithm for GP 2 (Long Version)}
\author{Ivaylo Hristakiev\sponsor{This author's work is supported by a PhD scholarship of the UK Engineering and Physical Sciences Research Council (EPSRC)} and Detlef Plump}
\institute{The University of York, UK}
\abstract{
The graph programming language GP 2 allows to apply sets of rule schemata (or ``attributed'' rules) non-deterministically. To analyse conflicts of programs statically, graphs labelled with expressions are overlayed to construct critical pairs of rule applications. Each overlay induces a system of equations whose solutions represent different conflicts. We present a rule-based unification algorithm for GP expressions that is terminating, sound and complete. For every input equation, the algorithm generates a finite set of substitutions. Soundness means that each of these substitutions solves the input equation. 
Since GP labels are lists constructed by concatenation, unification modulo associativity and unit law is required. This problem, which is also known as \emph{word unification}, is infinitary in general but becomes finitary due to GP's rule schema syntax and the assumption that rule schemata are left-linear. Our unification algorithm is complete in that every solution of an input equation is an instance of some substitution in the generated set. 
}

\input{macros}

\begin{document}
\maketitle
\input{introduction}
\input{rule-schemata}
\input{unification}
\input{completeness} 
\input{conclusion}
\bibliographystyle{eceasst}
\bibliography{abbr,%
              gragra,%
              gratra-languages,%
              unif}
\end{document}

%% file: frontmatter.tex
\volume{71}{2014} 
\volumetitle{
Graph Computation Models\\
Selected Revised Papers from GCM 2014}
\volumeshort{
Selected Revised Papers from GCM 2014}
\guesteds{
Rachid Echahed, Annegret Habel, Mohamed Mosbah}

%% file: macros.tex
\newtheorem{assumption}{Assumption}

{\begin{list}{$\bullet$}{\topsep\itemsep \parskip.5ex \partopsep0ex}}%
{\end{list}}

\newcounter{enum}
{\begin{list}{(\arabic{enum})}%
{\usecounter{enum}\topsep\itemsep \parskip.5ex \partopsep0ex}}%
{\end{list}}

{\begin{trivlist}\item[]%
{\textit{Proof sketch.}}}%
{\qed\end{trivlist}}

{\begin{trivlist}\item[]%
{\textit{Proof of #1.}}}%
{\qed\end{trivlist}}

{\begin{trivlist}\item[]%
{\textit{Acknowledgement.}}}%
{\qed\end{trivlist}}

\newcommand{\red}[1]{\textcolor{red}{#1}}
\newcommand{\blue}[1]{\textcolor{blue}{#1}}
\newcommand{\ttt}{\texttt}

\newcommand{\mylabelC}[2] { \left\{
 \begin{array}{l} #1 \\ \hline #2 \end{array} \right\} 
}

\newcommand{\meq}{=}
\newcommand{\meqq}{=^{?}}

\newcommand{\Z}{\mathbb{Z}}

\newcommand{\mtt}{\mathtt}
\newcommand{\mrm}{\mathrm}

\newcommand{\C}{\mathrm{Char}}

\newcommand{\G}{\mathcal{G}}

\newcommand{\tuple}[1]{\langle#1\rangle}

\newcommand{\dom}{\mathrm{Dom}}
\newcommand{\vran}{\mathrm{VRan}}

\newcommand{\dder}{\Rightarrow}
\newcommand{\ldder}{\Leftarrow}

\newcommand{\grey}{\graphnodecolour{.75}}
\newcommand{\dashed}{\graphlinedash{4 2}}

%% file: introduction.tex
\section{Introduction}
\label{sec:intro}

A common programming pattern in the graph programming language GP 2 \cite{Plump12a, Bak-Faulkner-Plump-Runciman15a} is to apply a set of graph transformation rules as long as possible. To execute such a loop $\{r_1,\dots,r_n\}!$ on a host graph, in each iteration an applicable rule $r_i$ is selected and applied. As rule selection and rule matching are non-deterministic, different graphs may result from the loop. Thus, if the programmer wants the loop to implement a function, a static analysis that establishes or refutes functional behaviour would be desirable.

The above loop is guaranteed to produce a unique result if the rule set $\{r_1,\dots,r_n\}$ is terminating and confluent. However, conventional confluence analysis via critical pairs \cite{Plump05a} assumes rules with constant labels whereas GP employs rule schemata (or ``attributed'' rules) whose graphs are labelled with expressions. Confluence of attributed graph transformation rules has been considered in \cite{Heckel-Kuester-Taentzer02a,Ehrig-Ehrig-Prange-Taentzer06a,Golas-Lambers-Ehrig-Orejas12a}, but we are not aware of \emph{algorithms}\/ that check confluence over non-trivial attribute algebras such as GP's which includes list concatenation and Peano arithmetic. The problem is that one cannot use syntactic unification (as in logic programming) when constructing critical pairs and checking their joinability, but has to take into account all equations valid in the attribute algebra.

For example, \cite{Heckel-Kuester-Taentzer02a} presents a method of constructing critical pairs in the case where the equational theory of the attribute algebra is represented by a confluent and terminating term rewriting system. The algorithm first computes normal forms of the attributes of overlayed nodes and subsequently constructs the most general (syntactic) unifier of the normal forms. This has been shown to be incomplete \cite[p.198]{Ehrig-Ehrig-Prange-Taentzer06a} in that the constructed set of critical pairs need not represent all possible conflicts. For, the most general unifier produces identical attributes---but it is necessary to find all substitutions that make attributes equivalent in the algebra's theory.

Graphs in GP rule schemata are labelled with lists of integer and string expressions, where lists are constructed by concatenation. In host graphs, list entries must be constant values. Integers and strings are subtypes of lists in that they represent lists of length one. 

As a simple example, consider the program in Figure~\ref{fig:distances} for calculating shortest distances. The program expects input graphs with non-negative integers as edge labels, and arbitrary lists as node labels. There must be a unique marked node (drawn shaded) whose shortest distance to each reachable node has to be calculated.
\begin{figure}[hbt]
\label{fig:distances}
\input{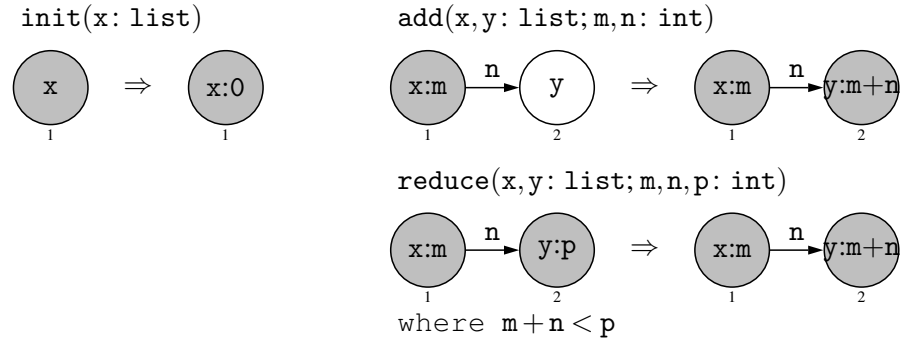}
\caption{A program calculating shortest distances}
\end{figure}
The rule schemata $\mtt{init}$ and $\mtt{add}$ append distances to the labels of nodes that have not been visited before, while $\mtt{reduce}$ decreases the distance of nodes that can be reached by a path that is shorter than the current distance.

To construct the conflicts of the rule schemata $\mtt{add}$ and $\mtt{reduce}$, their left-hand sides are overlayed. For example, the structure of the left-hand graph of $\mtt{reduce}$ can match the following structure in two different ways:
\begin{center}
\graphnodesize{1}
\grapharrowlength{.2}
\grapharrowwidth{.6}
\autodistance{1.0}
\graphnodecolour{1}

\begin{graph}(2.5,1.1)
\roundnode{1}(0.4,.6)[\grey]
\roundnode{2}(2.1,.6)[\grey]
\dirbow{1}{2}{.15}
\dirbow{2}{1}{.15}
\end{graph}
\end{center}
Consider a copy of $\mtt{reduce}$ in which the variables have been renamed to $\mtt{x'}$, $\mtt{m'}$, etc. To match $\mtt{reduce}$ and its copy differently requires solving the system of equations $\langle \mathtt{x{:}m} \meqq \mathtt{y'{:}p'},\, \mathtt{y{:}p} \meqq \mathtt{x'{:}m'}\rangle$. Solutions to these equations should be as general as possible to represent all potential conflicts resulting from the above overlay. In this simple example, it is clear that the substitution 
\[ \sigma = \{\mtt{x' \mapsto y},\,\mtt{m' \mapsto p},\,\mtt{y' \mapsto x},\,\mtt{p' \mapsto m}\} \]
is a most general solution. It gives rise to the following critical pair:\footnote{For simplicity, we ignore the condition of $\mtt{reduce}$.}
\begin{center}
\graphnodesize{1}
\grapharrowlength{.2}
\grapharrowwidth{.6}
\autodistance{1.0}
\graphnodecolour{1}

$\begin{array}{ccccc}
\begin{graph}(2.5,1.1)
\roundnode{A}(0.4,0.6)[\graphnodesize{1.1}\grey]
\autonodetext{A}{$\mathtt{x{:}p{+}n'}$}
\roundnode{B}(2.1,0.6)[\grey]
\autonodetext{B}{$\mathtt{y{:}p}$}
\autonodetext{A}[s]{\tiny 1}
\autonodetext{B}[s]{\tiny 2} 
\dirbow{A}{B}{.15}
\bowtext{A}{B}{.25}{$\mathtt{n}$}
\dirbow{B}{A}{.15}
\bowtext{B}{A}{.3}{$\mathtt{n'}$}
\end{graph}
&
\begin{graph}(.8,1.1)
\freetext(.4,.6){$\ldder$}
\end{graph}
&
\begin{graph}(2.5,1.1)
\roundnode{A}(0.4,0.6)[\grey]
\autonodetext{A}{$\mathtt{x{:}m}$}
\roundnode{B}(2.1,0.6)[\grey]
\autonodetext{B}{$\mathtt{y{:}p}$}
\autonodetext{A}[s]{\tiny 1}
\autonodetext{B}[s]{\tiny 2} 
\dirbow{A}{B}{.15}
\bowtext{A}{B}{.25}{$\mathtt{n}$}
\dirbow{B}{A}{.15}
\bowtext{B}{A}{.3}{$\mathtt{n'}$}
\end{graph}
&
\begin{graph}(.8,1.1)
\freetext(.4,.6){$\dder$}
\end{graph}
&
\begin{graph}(2.5,1.1)
\roundnode{A}(0.4,0.6)[\grey]
\autonodetext{A}{$\mathtt{x{:}m}$}
\roundnode{B}(2.1,0.6)[\grey]
\autonodetext{B}{$\mathtt{y{:}m{+}n}$}
\autonodetext{A}[s]{\tiny 1}
\autonodetext{B}[s]{\tiny 2} 
\dirbow{A}{B}{.15}
\bowtext{A}{B}{.25}{$\mathtt{n}$}
\dirbow{B}{A}{.15}
\bowtext{B}{A}{.3}{$\mathtt{n'}$}
\end{graph}
\end{array}$
\end{center}

In general though, equations can arise that have several independent solutions. For example, the equation $\langle \mtt{n{:}x} \meqq \mtt{y{:}2} \rangle$ (with $\mtt{n}$ of type $\mtt{int}$ and $\mtt{x}$,$\mtt{y}$ of type $\mtt{list}$) has the minimal solutions
\[ \sigma_1 = \{\mtt{x},\mtt{y} \mapsto \mtt{empty},\, \mtt{n \mapsto 2}\}\quad \text{and}\quad \sigma_2 = \{\mtt{x} \mapsto \mtt{z{:}2,\, \mtt{y} \mapsto \mtt{n{:}z}}\} \]
where $\mtt{empty}$ represents the empty list and $\mtt{z}$ is a list variable. 

Seen algebraically, we need to solve equations modulo the associativity and unit laws
\[ \mrm{AU} = \{x:(y:z) = (x:y):z,\; \mtt{empty}:x = x,\; x:\mtt{empty} = x\}. \]
This problem is similar to \emph{word unification}\/ \cite{Baader-Snyder01a}, which attempts to solve equations modulo associativity. 
(Some authors consider $AU$-unification as word unification, e.g. \cite{Jaffar90a}).
Solvability of word unification is decidable, albeit in PSPACE \cite{Plandowski99a}, but there is not always a finite complete set of solutions. The same holds for AU-unification (see Section~ \ref{sec:unification}). Fortunately, GP's syntax for left-hand sides of rule schemata forbids labels with more than one list variable. It turns out that by additionally forbidding shared list variables between different left-hand labels of a rule, rule overlays induce equation systems possessing finite complete sets of solutions.

This paper is the first step towards a static confluence analysis for GP programs. In Section~\ref{sec:unification}, we present a rule-based unification algorithm for equations between left-hand expressions of rule schemata. We show that the algorithm always terminates and that it is sound in that each substitution generated by the algorithm is an AU-unifier of the input equation. Moreover, the algorithm is complete in that every unifier of the input equation is an instance of some unifier in the computed set of solutions.

This paper is an extended version of the workshop paper \cite{Hristakiev-Plump15a}.

%% file: rule-schemata.tex
\section{GP Rule Schemata}
\label{sec:rule-schemata}

We refer to \cite{Plump12a, Bak-Faulkner-Plump-Runciman15a} for the definition of GP and more example programs. In this section, we define (unconditional) rule schemata which are the ``building blocks'' of graph programs. 

A \emph{graph}\/ over a label set $\mathcal{C}$ is a system $G=(V,E,s,t,l,m)$, where $V$ and $E$ are finite sets of \emph{nodes} (or \emph{vertices}) and \emph{edges}, $s,t\colon E\to V$ are the \emph{source} and \emph{target} functions for edges, $l\colon V \to \mathcal{C}$ is the node labelling function and $m\colon E\to \mathcal{C}$ is the edge labelling function. We write $\G(\mathcal{C})$ for the class of all graphs over $\mathcal{C}$. 

Figure~\ref{fig:ruleschema} shows an example for the declaration of a rule schema. 
\begin{figure}[htb]
 \begin{center}
  \input{ruleschema.graph}
 \end{center}
\caption{Declaration of a GP rule schema \label{fig:ruleschema}}
\end{figure}
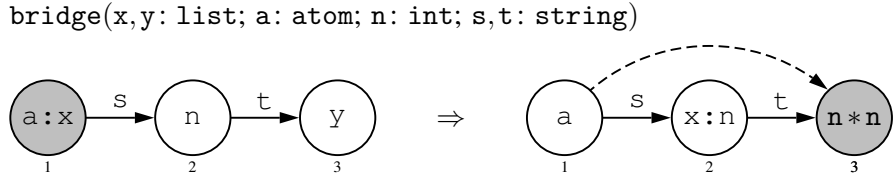
The types $\mtt{int}$ and $\mtt{string}$ represent integers and character strings. Type $\mtt{atom}$ is the union of $\mtt{int}$ and $\mtt{string}$, and $\mtt{list}$ represents lists of atoms. Given lists $l_1$ and $l_2$, we write $l_1\mathop{\mtt{:}}l_2$ for the concatenation of $l_1$ and $l_2$. The empty list is denoted by $\mtt{empty}$. In pictures of graphs, nodes or edges without label (such as the dashed edge in Figure~\ref{fig:ruleschema}) are implicitly labelled with the empty list. We equate lists of length one with their entry to obtain the syntactic and semantic \emph{subtype}\/ relationships shown in Figure~\ref{fig:subtypes}. For example, all labels in Figure~\ref{fig:ruleschema} are list expressions. 

\begin{figure}[!hb]
\input{subtypes.graph}
\caption{Subtype hierarchy for labels \label{fig:subtypes}}
\end{figure}
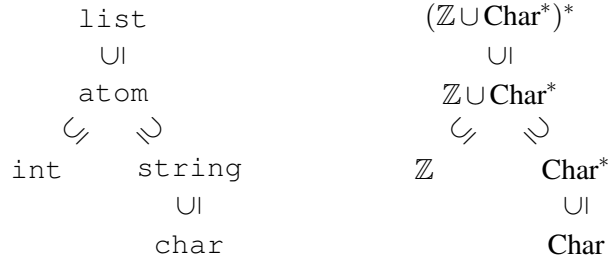

Figure~\ref{fig:abstract_label_syntax} gives a grammar in Extended Backus-Naur Form defining the abstract syntax of labels. The function \texttt{length} return the length of a variable, while \texttt{indeg} and \texttt{outdeg} access the indegree resp.\ outdegree of a left-hand node in the host graph. 

\begin{figure}[!htb] 
\begin{center}
\renewcommand{\arraystretch}{1.2}
\begin{tabular}{lcl}
Integer & ::= & Digit \{Digit\} $\mid$ IVar \\
&& $\mid$ `$-$' Integer $\mid$ Integer ArithOp Integer \\
&& $\mid$ \texttt{length} `(' LVar $\mid$ AVar $\mid$ SVar `)' \\
&& $\mid$ (\texttt{indeg} $\mid$ \texttt{outdeg}) `(' Node `)' \\
ArithOp & ::= & `\texttt{+}' $\mid$ `\texttt{-}' $\mid$ `$\ast$' $\mid$ `\texttt{/}' \\
String & ::= & `\,``\,' \{Char\} `\,''\,' $\mid$ SVar $\mid$ String `.' String\\
Atom & ::= & Integer $\mid$ String $\mid$ AVar \\
List & ::= & \texttt{empty} $\mid$ Atom $\mid$ LVar $\mid$ List `:' List \\
Label & ::= & List [Mark] \\
Mark & ::= & \texttt{red} $\mid$  \texttt{green} $\mid$ \texttt{blue} $\mid$  \texttt{grey} $\mid$ \texttt{dashed} $\mid$ \texttt{any}
\end{tabular}
\end{center}
\caption{Abstract syntax of rule schema labels \label{fig:abstract_label_syntax}}
\end{figure}

Figure~\ref{fig:abstract_label_syntax} defines four syntactic categories of expressions: Integer, String, Atom and List, where Integer and String are subsets of Atom which in turn is a subset of List. Category Node is the set of node identifiers used in rule schemata. Moreover, IVar, SVar, AVar and LVar are the sets of variables of type $\mtt{int}$, $\mtt{string}$, $\mtt{atom}$ and $\mtt{list}$. We assume that these sets are disjoint and define $\mrm{Var} = \mrm{IVar} \cup \mrm{SVar} \cup \mrm{AVar} \cup \mrm{LVar}$. The mark components of labels are represented graphically rather than textually. 

Each expression $l$\/ has a unique smallest type, denoted by $\mrm{type}(l)$, which can be read off the hierarchy in Figure~\ref{fig:subtypes} after $l$\/ has been normalised with the rewrite rules shown at the beginning of Subsection~\ref{subsec:unif-algorithm}. We write $\mrm{type}(l_1) < \mrm{type}(l_2)$ or $\mrm{type}(l_1) \leq \mrm{type}(l_2)$ to compare types according to the subtype hierarchy. If the types of $l_1$ and $l_2$ are incomparable, we write $\mrm{type}(l_1) \mathop{\|} \mrm{type}(l_2)$.

The values of rule schema variables at execution time are determined by graph matching. To ensure that matches induce unique ``actual parameters'', expressions in the left graph of a rule schema must have a simple shape.

\begin{definition}[Simple expression]
\label{def:simple_label}
\normalfont
A \emph{simple}\/ expression contains no arithmetic operators, no length or degree operators, no string concatenation, and at most one occurrence of a list variable.
\end{definition}
In other words, simple expressions contain no unary or binary operators except list concatenation, and at most one occurrence of a list variable. For example, given the variable declarations of Figure~\ref{fig:ruleschema}, \texttt{a:x} and \texttt{y:n:n} are simple expressions whereas $\mtt{n \ast 2}$ or \texttt{x:y} are not simple.

Our definition of simple expressions is more restrictive than that in \cite{Plump12a} because we exclude string concatenation and the unary minus. These operations (especially string concatenation) would considerably inflate the unification algorithm and its completeness proof, without posing a substantial challenge.

\begin{definition}[Rule schema]
\label{def:ruleschema}
\normalfont
A \emph{rule schema} $\tuple{L,\, R,\, I}$ consists of graphs $L,R$\/ in $\G(\mrm{Label})$ and a set $I$, the \emph{interface}, such that $I = V_L \cap V_R$. All labels in $L$\/ must be simple and all variables occurring in $R$\/ must also occur in $L$. 
\end{definition}

When a rule schema is graphically declared, as in Figure \ref{fig:ruleschema}, the interface $I$\/ is represented by the node numbers in $L$\/ and $R$. Nodes without numbers in $L$\/ are to be deleted and nodes without numbers in $R$\/ are to be created. All variables in $R$\/ have to occur in $L$\/ so that for a given match of $L$\/ in a host graph, applying the rule schema produces a graph that is unique up to isomorphism.

\begin{assumption}[Left-linearity] We assume that rule schemata $\tuple{L,\, R,\, I}$ are \emph{left-linear}, that is, the labels of different nodes or edges in $L$\/ do not contain the same list variable.
\end{assumption}

This assumption is necessary to ensure that the solutions of the equations resulting from overlaying two rule schemata can be represented by a finite set of unifiers. For example, without this assumption it is easy to construct two rule schemata that induce the system of equations $\langle \mtt{x:1 \meqq y,\; y \meqq 1:x}\rangle$. This system has solutions $\{\mtt{x \mapsto empty, y \mapsto 1}\}$, $\{\mtt{x \mapsto 1, y \mapsto 1:1}\}$, \newline $\{\mtt{x \mapsto 1:1, y \mapsto 1:1:1}\},\ldots$  which form a infnite, minimal and compete set of solutions (See Definition \ref{def:completeset} below).

%% file: ruleschema.graph
   \setlength{\unitlength}{0.5in}
   \graphnodesize{0.8}
   \grapharrowlength{0.2}
   \grapharrowwidth{0.6}
   \autodistance{1.0}
   \graphnodecolour{1}
\begin{tabular}{l}
$\mtt{bridge (x,y\colon list;\; a\colon atom;\; n\colon int;\; s,t\colon string)}$
\\
\begin{graph}(3.8,1.4)(.2,1.17)
 \roundnode{n1}(.6,1.7)[\grey]
  \autonodetext{n1}{\texttt{a:x}}\autonodetext{n1}[s]{\tiny 1}
 \roundnode{n2}(2.1,1.7)
  \autonodetext{n2}{\texttt{n}}\autonodetext{n2}[s]{\tiny 2}
 \roundnode{n3}(3.6,1.7)
  \autonodetext{n3}{\texttt{y}}\autonodetext{n3}[s]{\tiny 3}
 \diredge{n1}{n2}\bowtext{n1}{n2}{.12}{\texttt{s}}
 \diredge{n2}{n3}\bowtext{n2}{n3}{.12}{\texttt{t}}
\end{graph}
\begin{graph}(1.4,1.4)
 \freetext(.7,.5){$\dder$}
\end{graph}
\begin{graph}(3.8,1.4)(.2,1.17)
 \roundnode{n1}(.6,1.7)
  \autonodetext{n1}{\texttt{a}}\autonodetext{n1}[s]{\tiny 1}
  \roundnode{n2}(2.1,1.7)
  \autonodetext{n2}{\texttt{x:n}}\autonodetext{n2}[s]{\tiny 2}
 \roundnode{n3}(3.6,1.7)[\grey]\autonodetext{n3}[s]{\tiny 3}
  \autonodetext{n3}{$\mtt{n \ast n}$}\autonodetext{n3}[s]{\tiny 3}
 \dirbow{n1}{n3}{.25}[\dashed]
 \diredge{n1}{n2}\bowtext{n1}{n2}{.12}{\texttt{s}}
 \diredge{n2}{n3}\bowtext{n2}{n3}{.12}{\texttt{t}}
\end{graph}
\end{tabular}

%% file: subtypes.graph
   \graphnodesize{0.3}
   \graphlinewidth{0.01}
   \grapharrowlength{0.1}
   \grapharrowwidth{0.5}

\begin{center}
\begin{tabular}{ccc}
\begin{graph}(3.25,3.6)(0,-.8)
\opaquetextfalse
\freetext(1.5,2.5){\texttt{list}}
\freetext(1.5,1.5){\texttt{atom}}
\freetext(.5,.5){\texttt{int}}
\freetext(2.5,.5){\texttt{string}}
\freetext(2.5,-.5){\texttt{char}}
\freetext(1.5,2){\begin{turn}{90} \mbox{$\subseteq$} \end{turn}}
\freetext(1,1){\begin{turn}{45} \mbox{$\subseteq$} \end{turn}}
\freetext(2,1){\begin{turn}{-45} \mbox{$\supseteq$} \end{turn}}
\freetext(2.5,0){\begin{turn}{90} \mbox{$\subseteq$} \end{turn}}
\end{graph}
& \hspace{1cm} &
\begin{graph}(3.25,3.6)(0,-.8)
\opaquetextfalse
\freetext(1.5,2.5){$(\Z \cup \C^*)^*$}
\freetext(1.5,1.5){$\Z \cup \C^*$}
\freetext(.5,.5){$\Z$}
\freetext(2.5,.5){$\C^*$}
\freetext(2.5,-.5){$\C$}
\freetext(1.5,2){\begin{turn}{90} \mbox{$\subseteq$} \end{turn}}
\freetext(1,1){\begin{turn}{45} \mbox{$\subseteq$} \end{turn}}
\freetext(2,1){\begin{turn}{-45} \mbox{$\supseteq$} \end{turn}}
\freetext(2.5,0){\begin{turn}{90} \mbox{$\subseteq$} \end{turn}}
\end{graph}
\end{tabular}
\end{center}

%% file: unification.tex
\section{Unification}
\label{sec:unification}

We start with introducing some technical notions such as substitutions, unification problems and complete sets of unifiers. Then, in Subsection \ref{subsec:unif-algorithm}, we present our unification algorithm. In Subsection \ref{subsec:unif-properties}, we prove that the algorithm terminates and is sound.

\input{Background}

\input{Prelim}
\input{Rules}
\input{Algorithm}
\input{Termination}

\input{CC}

\input{Completeness}

%% file: Background.tex
\subsection{Preliminaries}
\label{subsec:preliminaries}

A \emph{substitution}\/ is a family of mappings $\sigma = (\sigma_X)_{X \in \{I,S,A,L\}}$ where $\sigma_I\colon \mrm{IVar} \to \mrm{Integer}$, $\sigma_S\colon \mrm{SVar} \to \mrm{String}$, $\sigma_A\colon \mrm{AVar} \to \mrm{Atom}$, $\sigma_L\colon \mrm{LVar} \to \mrm{List}$. Here Integer, String, Atom and List are the sets of expressions defined by the GP label grammar of Figure \ref{fig:abstract_label_syntax}. For example, if $\mtt{z} \in \mrm{LVar}$,  $\mtt{x} \in \mrm{IVar}$ and $\mtt{y} \in \mrm{SVar}$, then we write $\sigma = \{\mtt{x} \mapsto \mtt{x+1},\, \mtt{z} \mapsto \mtt{y:-x:y}\}$ for the substitution that maps $\mtt{x}$ to $\mtt{x+1}$, $\mtt{z}$ to $\mtt{y:-x:y}$ and every other variable to itself.

Applying a substitution $\sigma$ to an expression $t$, denoted by $t\sigma$, means to replace every variable $x$ in $t$ by $\sigma(x)$ simultaneously. In the above example, $(\mtt{z:-x})\sigma = \mtt{y:-x:y:-(x+1)}$. 

By $\dom(\sigma)$ we denote the set $\{x \in \mrm{Var} \mid \sigma(x) \neq x\}$ and by $\vran(\sigma)$ the set of variables occurring in the expressions $\{\sigma(x) \mid x \in \mrm{Var}\}$.
A substitution $\sigma$ is \emph{idempotent}\/ if $\dom(\sigma)\mathop{\cap}\vran(\sigma)~=~\emptyset$.
The \emph{composition} of two substitutions $\sigma$ and $\theta$, is defined as $ x (\sigma \circ \theta) = \begin{cases}
 (x\sigma) \theta & \text{if~}x \in \dom(\sigma) \\
 x\theta & \text{otherwise}
\end{cases}$
and is an associative operation.

\begin{definition}[Unification problem]
\normalfont
A \emph{unification problem} is a pair of an equation and a substitution
\[ P = \langle s \meqq t,\sigma_P \rangle\]
where $s$ and $t$ are simple list expressions without common variables.
\end{definition}
The symbol $=^?$ signifies that the equation must be \emph{solved}\/ rather than having to hold for all values of variables. The purpose of $\sigma_P$ is  for the unification algorithm (Section \ref{sec:unifAlgorithm}) to record a partial solution. An illustration of this concept will be seen in \autoref{Unifexample}.

In \autoref{sec:rule-schemata}, we already assumed that GP rule schemata need to be left-linear. Now, the problem of solving a system of equations $\{s_1 = t_1, s_2 = t_2\}$ can be broken down to solving individual equations and combining the answers - if $\sigma_1$ and $\sigma_2$ are solutions to each individual equation, then $\sigma_1 \cup \sigma_2$ is a solution to the combined problem as $\sigma_1$ and $\sigma_2$ do not share variables.

Consider the equational axioms for associativity and unity,
\[ \mrm{AU} = \{\mtt{x:(y:z)} = \mtt{(x:y):z},\; \mtt{empty:x} = \mtt{x},\; \mtt{x:empty} = \mtt{x}\} \]
where $\mtt{x},\mtt{y},\mtt{z}$ are variables of type \texttt{list}, and let $=_{\mrm{AU}}$ be the equivalence relation on expressions generated by these axioms. 

\begin{definition}[Unifier]
\normalfont
Given a unification problem $P =  \langle s \meqq t,\sigma_P \rangle$ a \emph{unifier}\/ of $P$ a is a substitution $\theta$ such that 
\[s \theta =_{\mrm{AU}} t \theta \text{~and~}x_i \theta =_{\mrm{AU}} t_i \theta \]
for each binding $\{\mtt{x_i \mapsto t_i}\}$ in $\sigma_P$\enspace.
\end{definition}


The set of all unifiers of $P$ is denoted by $\mathcal{U}(P)$. We say that $P$ is \emph{unifiable} if $\mathcal{U}(P) \neq \emptyset$.

The special unification problem $\mathtt{fail}$ represents failure and has no unifiers. A problem $P~=~ \langle s \meqq t,\varnothing \rangle$ is \emph{initial} and $P = \langle \varnothing,\sigma_P \rangle$ is \emph{solved}.

A substitution $\sigma$ is \emph{more general}\/ on a set of variables $X$ than a substitution $\theta$ if there exists a substitution $\lambda$ such that $x\theta =_{\mrm{AU}} x\sigma\lambda $ for all $x \in X$. 
In this case we write $\sigma \leqq_X \theta$ and say that $\theta$ is an \emph{instance} of $\sigma$ on $X$. Substitutions $\sigma$ and $\theta$ are \emph{equivalent}\/ on $X$, denoted by $\sigma=_X\theta$, if $\sigma \leqq_X \theta$ and $\theta \leqq_X \sigma$. 

\begin{definition}[Complete set of unifiers \cite{Plotkin72}]
\label{def:completeset}
\normalfont
A set $\mathcal{C}$ of substitutions is a \emph{complete set of unifiers} of a unification problem $P$ if
\begin{enumerate}
\item (Soundness) $\mathcal{C} \subseteq \mathcal{U}(P)$, that is, each substitution in $\mathcal{C}$ is a unifier of $P$, and
\item (Completeness) for each $\theta \in \mathcal{U}(P)$ there exists $\sigma \in \mathcal{C}$ such that $\sigma \leqq_X \theta$, where $X~=~\mrm{Var}(P)$ is the set of variables occurring in $P$.
\end{enumerate}	
Set $\mathcal{C}$ is also \emph{minimal} if each pair of distinct unifiers in $\mathcal{C}$ are incomparable with respect to $\leqq_X$.
\end{definition}

If a unification problem $P$\/ is not unifiable, then the empty set $\varnothing$ is a minimal complete set of unifiers of $P$. 

%
For simplicity, we replace $=^?$ with $=$ in unification problems from now on.

\begin{example}
\normalfont
The minimal complete set of unifiers of the problem $\langle \mtt{a:x}=\mtt{y:2}\rangle$ (where \texttt{a} is an atom variable and \texttt{x},\texttt{y} are list variables) is $\{\sigma_1,\sigma_2\}$ with
\[
\sigma_1 = \{\mtt{a} \mapsto 2,\, \mtt{x} \mapsto \mtt{empty},\, \mtt{y} \mapsto \mtt{empty}\}\quad \text{and}\quad
\sigma_2 = \{\mtt{x} \mapsto \mtt{z:2},\, \mtt{y} \mapsto \mtt{a:z}\}.
\]
We have $\sigma_1 (\mtt{a:x}) = \mtt{2:empty} =_{\mrm{AU}} \mtt{2} =_{\mrm{AU}} \mtt{empty:2} = \sigma_1 (\mtt{y:2})$ and $\sigma_2 (\mtt{a:x}) = \mtt{a:z:2} = \sigma_2 (\mtt{y:2})$. Other unifiers such as $\sigma_3 = \{\mtt{x} \mapsto 2,\, \mtt{y} \mapsto \mtt{a}\}$ are instances of $\sigma_2$.
\end{example}

%% file: Prelim.tex
\subsection{Unification Algorithm}
\label{subsec:unif-algorithm}

We start with some notational conventions for the rest of this section:
\begin{itemize}
\item $L,M$\/ stand for simple expressions,
\item $\mtt{x,y,z}$ stand for variables of any type (unless otherwise specified),
\item $\mtt{a,b}$ stand for 
\begin{enumerate}[(i)]
\item simple string or integer expressions, or
\item string, integer or atom variables
\end{enumerate}
\item $\mtt{s,t}$ stand for 
\begin{enumerate}[(i)]
\item simple string or integer expressions, or
\item variables of any type
\end{enumerate}
\end{itemize}

\paragraph{\normalfont\textbf{Preprocessing.}}
Given a unification problem $P = \langle s \meqq t,\sigma \rangle$, we rewrite the terms in $s$ and $t$ using the reduction rules
\[ L:\mtt{empty} \to L\quad \text{and}\quad \mtt{empty}:L \to L \]
where $L$ ranges over list expressions. These reduction rules are applied exhaustively before any of the transformation rules. For example, 
\[ \mtt{x:empty:1:empty} \to \mtt{x:1:empty} \to \mtt{x:1}. \]
We call this process \emph{normalization}. In addition, the rules are applied to each instance of a transformation rule (that is, once the formal parameters have been replaced with actual parameters) before it is applied, and also after each transformation rule application. 

%% file: Rules.tex
\paragraph{\textbf{Transformation rules.}}

Figure \ref{rules} shows the transformation rules, the essence of our approach, in an inference system style where each rule consists of a premise and a conclusion.

\vspace{\baselineskip}
\begin{tabular}{lp{10cm}}
\textsf{Remove}: & deletes trivial equations \\
\textsf{Decomp1}: &  syntactically equates a list variable with an atom expression or list variable \\
\textsf{Decomp1'}: &  syntactically equates an atom variable with an expression of lesser type \\
\textsf{Decomp2/2'}: &  assigns a list variable to start with another list variable  \\
\textsf{Decomp3}: &  removes identical symbols from the head  \\
\textsf{Decomp4}: &  solves an atom variable  \\
\textsf{Subst1}: & solves a variable\\
\textsf{Subst2}: & assigns $\mtt{empty}$ to a list variable \\
\textsf{Subst3}: & assigns an atom prefix to a list variable \\
\textsf{Orient1/2}: & moves variables to left-hand side \\
\textsf{Orient3}: & moves variables of larger type to left-hand side \\
\textsf{Orient4}: &  moves a list variable to the left-hand side  \\
\end{tabular}
\vspace{\baselineskip}

%
\setlength{\footskip}{40pt}
\begin{figure}[!p]
\setstretch{3.6}

	\infer[\mathsf{Remove}]
	{ \langle \varnothing, \sigma \rangle}
	{ \langle L = L, \sigma \rangle }
		
	\infer[\mathsf{Decomp1}]
	{\langle L = M,~\sigma \circ \{\mtt{x \mapsto s}\}\rangle}
	{\langle \mtt{x}:L = \mtt{s}:M, \sigma \rangle
	& L \neq \mathtt{empty} 
	& \mathrm{type}(\mtt{x}) = \ttt{list}
	}
	
	\infer[\mathsf{Decomp1'}]
	{\langle (L = M) \{\mtt{x \mapsto a}\} ,~\sigma \circ \{\mtt{x \mapsto a}\} \rangle}
	{\langle \mtt{x}:L = \mtt{a}:M , \sigma \rangle
	& L \neq \mathtt{empty} 
	& \mathrm{type}(\mtt{a}) \leq \mathrm{type}(\mtt{x}) \leq \ttt{atom}
	}	 
	 	
	\infer[\mathsf{Decomp2}]
	{\langle \mtt{x'}:L = M , \sigma \circ \{ \mtt{x \mapsto y:x'}\} \rangle
	}
	{\langle \mtt{x}:L = \mtt{y}:M , \sigma \rangle
	& L \neq \mathtt{empty} 
	& \mathrm{type}(\mtt{x}) = \ttt{list}
	& \mathrm{type}(\mtt{y}) = \ttt{list}
	& \mtt{x'}\text{ is a fresh list variable} 
	}	
	 
    \infer[\mathsf{Decomp2'}]
	{\langle L = \mtt{y'}:M , \sigma \circ \{ \mtt{y \mapsto x:y'}\} \rangle
	}
	{\langle \mtt{x}:L = \mtt{y}:M , \sigma \rangle
	& L \neq \mathtt{empty} 
	& \mathrm{type}(\mtt{x}) = \ttt{list}
	& \mathrm{type}(\mtt{y}) = \ttt{list}
	& \mtt{y'}\text{ is a fresh list variable} 
	}	
		
	\begin{tabular}{lr}    
	\infer[\mathsf{Decomp3}]
	{\langle L = M, \sigma \rangle
	}
	{\langle \mtt{s}:L = \mtt{s}:M, \sigma \rangle
	}	 & 	
	\infer[\mathsf{Decomp4}]
	{\langle \mtt{empty=y} , \sigma \circ \{ \mtt{x \mapsto a } \} \rangle
	}
	{\langle \mtt{x=a:y}, \sigma \rangle
	& \mathrm{type(x)} = \ttt{atom}
	& \mathrm{type(y)} = \ttt{list}
	}
	\end{tabular}	
	
	\infer[\mathsf{Subst1}]
	{\langle \varnothing , \sigma \circ \{\mtt{x} \mapsto L\} \rangle
	}
	{\langle \mtt{x} = L , \sigma \rangle
	& \mtt{x} \notin \mathrm{Var}(L) 
	&  \mathrm{type}(\mtt{x}) \geq \mathrm{type}(L) \footnotetext{$List \geq Atom$, $Atom \geq String$, $Atom \geq Int$}
	}	
	  
	\infer[\mathsf{Subst2}]
	{\langle L = M, \sigma \circ \{\mtt{x} \mapsto \texttt{empty}\} \rangle
	}
	{\langle \mtt{x}:L = M, \sigma \rangle
	& L \neq \mathtt{empty} 
	& \mathrm{type}(\mtt{x}) = \ttt{list}
	}
	 
	\infer[\mathsf{Subst3}]
	{\langle \mtt{x'}:L = M, \sigma \circ \{\mtt{x \mapsto a:x'}\} \rangle 
	}
	{\langle \mtt{x}:L = \mtt{a}:M\}, \sigma \rangle
	& L \neq \mathtt{empty} 
	& \mtt{x'}\text{ is a fresh list variable} 
	& \mathrm{type}(\mtt{x}) = \ttt{list}}
	
	\infer[\mathsf{Orient1}]
	{\langle \mtt{x}:L = \mtt{a}:M, \sigma \rangle }
	{\langle \mtt{a}:M = \mtt{x}:L, \sigma \rangle
	& \mtt{a}\text{ is not a variable}
	}

	\infer[\mathsf{Orient2}]
	{\langle \mtt{y} = \mtt{x}:L, \sigma \rangle
	}
	{\langle \mtt{x}:L = \mtt{y}, \sigma \rangle
	& L \neq \ttt{empty} 
	& \mathrm{type}(\mtt{x}) = \mathrm{type}(\mtt{y})
	} 
	
	\begin{tabular}{lr}
    \infer[\mathsf{Orient3}]
	{\langle \mtt{x}:L = \mtt{y}:M, \sigma \rangle
	}
	{\langle \mtt{y}:M = \mtt{x}:L, \sigma \rangle
	& \mathrm{type}(\mtt{y}) < \mathrm{type}(\mtt{x})
	} & 
	
	\infer[\mathsf{Orient4}]
	{\langle \mtt{x}:L = \ttt{empty}, \sigma \rangle
	}
	{\langle \ttt{empty} = \mtt{x}:L, \sigma \rangle
	& \mathrm{type}(\mtt{x}) = \ttt{list}
	}
	\end{tabular}		
	
\caption{Transformation rules}
\label{rules}
\setstretch{1}
\end{figure}
	
The rules induce a transformation relation $\Rightarrow$ on unification problems. In order to apply any of the rules to a problem $P$, the problem part of its premise needs to be \emph{matched} onto $P$. 
Subsequently, the boolean condition of the premise is checked and the rule \emph{instance} is normalized so that its premise is identical to $P$. 

For example, the rule \textsf{Orient3} can be matched to $P = \langle \mtt{a:2} = \mtt{m}, \;\sigma \rangle$ (where $\mtt{a}$ and $\mtt{m}$ are variables of type $\mtt{atom}$ and $\mtt{list}$, respectively) by setting $\mtt{y} \mapsto \mtt{a}$, $\mtt{x} \mapsto \mtt{m}$, $M \mapsto 2$, and $L \mapsto \mtt{empty}$. 
The rule instance and its normal form are then

\begin{tabular}{lcl}
&~&~\\
\infer{\langle \mtt{m}:\mtt{empty} = \mtt{a}:2, \sigma\rangle}{\langle \mtt{a}:2 = \mtt{m}:\mathtt{empty} , \sigma\rangle}
& \raisebox{2mm}{and} &
\infer{\langle\mtt{m} = \mtt{a}:2 , \sigma\rangle}{\langle\mtt{a}:2 = \mtt{m} , \sigma\rangle} \\
&~&~
\end{tabular}

where the conclusion of the normal form is the result of applying \textsf{Orient3} to $P$.

\newpage
Showing that a unification problem has no solution can be a lengthy affair because we need to compute all normal forms with respect to $\Rightarrow$. Instead, the  rules \textsf{Occur} and \textsf{Clash1} to \textsf{Clash4}, shown in \autoref{failure}, introduce \emph{failure}. Failure cuts off parts of the search tree for a given problem $P$. This is because if $P \dder \mrm{fail}$, then $P$ has no unifiers and it is not necessary to compute another normal form. Effectively, the failure rules have precedence over the other rules. They are justified by the following lemmata.


\begin{figure}[t]
\setstretch{3.6}

	\begin{tabular}{lr}
	\infer[\mathsf{Occur}]{\mrm{fail}}
	{\langle \mtt{x} = L,\sigma \rangle
	& \mtt{x} \in \mrm{Var}(L) 
	& \mtt{x} \neq L 
	& \mathrm{type}(\mtt{x}) = \ttt{list}
	} 
	&
	\infer[\mathsf{Clash2}]{\mathrm{fail}}
	{\langle \mtt{a}:L = \ttt{empty},\sigma \rangle 
	} \\
	
	\infer[\mathsf{Clash1}]{\mathrm{fail}}
	{\langle \mtt{a}:L = \mtt{b}:M,\sigma \rangle
	& \mtt{a} \neq \mtt{b} 
	& \mathrm{Var}(\mtt{a}) = \emptyset = \mathrm{Var}(\mtt{b})
	}
	 
	& 
	\infer[\mathsf{Clash3}]{\mathrm{fail}}
	{\langle \ttt{empty} = \mtt{a}:L ,\sigma \rangle 
	} \\
	
	\infer[\mathsf{Clash4}]{\mathrm{fail}}
	{\langle \mtt{x} = L, \sigma \rangle
	& \mathrm{type}(\mtt{x}) \mathop{\|} \mathrm{type}(L)} 
	\end{tabular}

\caption{Failure rules}
\label{failure}
\setstretch{1}
\end{figure}

\begin{lemma}
\label{1.6}
A normalised equation $\mtt{x} \meq L$ with $\mtt{x} \neq L$ has no solution if $L$ is a simple expression, $\mtt{x} \in Var(L)$ and $\mathrm{type}(\mtt{x}) = \ttt{list}$.
\end{lemma}
\begin{proof}
Since $\mtt{x} \in Var(L)$ and $\mtt{x} \neq L$, $L$ is of the form $s_1 : s_2 : \ldots : s_n$ with $n \geq 2$ and $\mtt{x} \in Var(s_i)$ for some $1 \leq i \leq n$. As $L$\/ is normalised, none of the terms $s_i$ contains the constant $\ttt{empty}$.
Also, since $L$ is simple, it contains no list variables other than $\mtt{x}$ and $\mtt{x}$ is not repeated. It follows $\sigma (\mtt{x}) \not=_{\mrm{AU}} \sigma (L)$ for every substitution $\sigma$.
\end{proof}

\begin{lemma}
\label{1.7}
Equations of the form $\mtt{a}:L = \ttt{empty}$ or $\ttt{empty} = \mtt{a}:L$ have no solution if $\mtt{a}$ is an atom expression.
\end{lemma}

\begin{lemma}
\label{1.8}
An equation $\mtt{a}:L = \mtt{b}:M$ with $\mtt{a} \neq \mtt{b}$ has no solution if $\mtt{a}$ and $\mtt{b}$ are atom expressions without variables.
\end{lemma}

%% file: Algorithm.tex
\paragraph{\normalfont\textbf{The algorithm.}} 
\label{sec:unifAlgorithm}
The unification algorithm in Figure \ref{alg} starts by normalizing the input equation, as explained above. It uses a queue of unification problems to search the derivation tree of $P$ with respect to $\dder$ in a breadth-first manner. The first step is to put the normalized problem $P$ on the queue.

The variable $\mathtt{next}$ holds the head of the queue. If $\mathtt{next}$ is in the form $\langle \varnothing, \sigma \rangle$, then $\sigma$ is a unifier of the original problem and is added to the set \texttt{U} of solutions. Otherwise, the next step is to construct all problems $P'$ such that $\mathtt{next} \Rightarrow P'$. If $P'$ is $\mtt{fail}$, then the derivation tree below $\mathtt{next}$ is ignored, otherwise $P'$ gets normalized and enqueued.

\begin{figure}[tb]
\begin{center}
$\begin{array}{ll}
\textsf{Unify}(\mtt{P}):~~ &  \mtt{U} := \emptyset\\
& \text{create empty queue $\mtt{Q}$ of unification problems} \\
& \text{normalize $\mtt{P}$}\\
& \mtt{Q.enqueue(\langle P,\varnothing\rangle)} \\
& \ttt{while~} \mtt{Q} \text{ is not empty } \\
& \;\;\;\mathtt{next := Q.dequeue()}\\
& \;\;\;\text{\ttt{if~}}\mathtt{next} \text{~is in the form~} \langle\varnothing,\sigma \rangle \\ 
& \;\;\;\;\;\;\mtt{U := U \cup \{ \sigma \}}\\
& \;\;\;\text{\ttt{else} \ttt{if} $\mtt{next \nRightarrow fail}$}\\
& \;\;\;\;\;\;\text{\ttt{foreach} $\mtt{P'}$ such that $\mtt{next} \Rightarrow \mtt{P'}$}\\
& \;\;\;\;\;\;\;\;\;\;\;\; \text{normalize $\mtt{P'}$} \\
& \;\;\;\;\;\;\;\;\;\;\;\; \mtt{Q.enqueue(P')} \\
& \;\;\;\;\;\;\text{\ttt{end} \ttt{foreach}} \\
& \;\;\;\;\;\;\text{\ttt{end} \ttt{if}} \\
& \;\;\;\text{\ttt{end} \ttt{if}} \\
& \text{\ttt{end while}} \\
& \text{\ttt{return} $\mtt{U}$} \\
& \\
\end{array}$
\caption{Unification algorithm}
\label{alg}
\end{center}
\end{figure}

\newpage
An example tree traversed by the algorithm is shown in Figure \ref{Unifexample}. Nodes are labelled with unification problems and edges represent applications of transformation rules. The root of the tree is the problem $\langle \mtt{y:2=a:x} \rangle$ to which the rules \textsf{Decomp1}, \textsf{Subst2} and \textsf{Subst3} can be applied. The three resulting problems form the second level of the search tree and are processed in turn. Eventually, the unifiers
\[ 
\begin{array}{rcl}
\sigma_1 & = & \{\mtt{x} \mapsto 2,\, \mtt{y} \mapsto \mtt{a}\} \\
\sigma_2 & = & \{\mtt{x} \mapsto \mtt{y':2},\, \mtt{y} \mapsto \mtt{a:y'}\} \\
\sigma_3 & = & \{\mtt{a} \mapsto 2,\, \mtt{x} \mapsto \mtt{empty},\, \mtt{y} \mapsto \mtt{empty}\}
\end{array}
\]
are found, which represent a complete set of unifiers of the initial problem. Note that the set is not minimal because $\sigma_1$ is an instance of $\sigma_2$. 

\begin{figure}[!p]
\begin{tikzpicture}[every node/.style = {align=center}, every edge/.style = {->}, level 1/.style={sibling distance = 4.5cm, level distance = 2.6cm, edge from parent/.style={draw,-latex}
},level 2/.style={sibling distance = 3.5cm, level distance = 3.2cm, edge from parent/.style={draw,-latex}
},level 3/.style={sibling distance = 3.5cm, level distance = 3.4cm, edge from parent/.style={draw,-latex}
},level 4/.style={sibling distance = 3.5cm, level distance = 3.4cm, edge from parent/.style={draw,-latex}
}]

\node (o) {$\mylabelC{\mathsf{y}:2  = \mathsf{a} :\mathsf{x} }{\varnothing}$}
	child {
	      node {$\mylabelC{2  = \mathsf{x}}{ \mathsf{y} \mapsto \mathsf{a}}$}
		      child {
		         node {$\mylabelC{\mathsf{x}  = 2}{ \mathsf{y} \mapsto \mathsf{a}}$}
		              child {
		              	    node {$\mylabelC{\varnothing}{ \mathsf{y} \mapsto \mathsf{a} \\ \mathsf{x} \mapsto 2 }$\\ \blue{solved}}
		              	    edge from parent node[left]{\textsf{Subst1}}
		              }
		         edge from parent node[left]{\textsf{Orient1}}
		      }
		  edge from parent node[left]{\textsf{Decomp1}\hspace*{1em}}
	      }
	child {
		node[xshift=-1cm] {$\mylabelC{2 = \mathsf{a} :\mathsf{x}}{ \mathsf{y} \mapsto \mtt{empty}}$}
			child {
				node {$\mylabelC{\mathsf{a} :\mathsf{x} = 2}{ \mathsf{y} \mapsto \mtt{empty}}$}
					child {
						node {$\mylabelC{\mathsf{x} = \mtt{empty}}{ \mathsf{y} \mapsto \mtt{empty} \\ \mathsf{a} \mapsto 2}$}
							child {
								node {$\mylabelC{\varnothing}{ \mathsf{y} \mapsto \mtt{empty} \\ \mathsf{a} \mapsto 2 \\ \mathsf{x} \mapsto \mtt{empty}}$ \\ \blue{solved}}
								edge from parent node[left]{$\mathsf{Subst1}$}
							}
						edge from parent node[left]{$\mathsf{Decomp1'}$}
					}	              
				edge from parent node[left]{$\mathsf{Orient1}$}
			} 
		edge from parent node[right]{\hspace*{.3em}$\mathsf{Subst2}$}
	}       
	child {
		node[xshift=1cm] {$\mylabelC{\mathsf{y'}:2 = \mathsf{x} }{\mathsf{y} \mapsto \mathsf{a}:\mathsf{y'}}$} 
			child {
				node {$\mylabelC{\mathsf{x} = \mathsf{y'}:2 }{\mathsf{y} \mapsto \mathsf{a}:\mathsf{y'}}$} 
					child {
						node {$\mylabelC{\varnothing}{\mathsf{y} \mapsto \mathsf{a}:\mathsf{y'} \\ \mathsf{x} \mapsto \mathsf{y'}:2}$ \\ \blue{solved}} 
						edge from parent node[left]{$\mathsf{Subst1}$}
					}
				edge from parent node[left]{$\mathsf{Orient2}$}
			} 
			child {
				node {$\mylabelC{\mathsf{y''}:2 = \mtt{empty} }{\mathsf{y} \mapsto \mathsf{a}:\mathsf{x:y''}}$}
					child {
						node {$\mylabelC{2 = \mtt{empty} }{\mathsf{y} \mapsto \mathsf{a}:\mathsf{x}}$}
							child {
								node[yshift=1cm] {\red{fail}}
								edge from parent node[right]{$\mathsf{Clash2}$}						
							}
						edge from parent node[right]{$\mathsf{Subst2}$}
					}                           
				edge from parent node[right]{$\mathsf{Decomp2}$}
			} 
			child {
				node {$\mylabelC{2 = \mathsf{x'} }{\mathsf{y} \mapsto \mathsf{a}:\mathsf{y'} \\ \mathsf{x} \mapsto \mathsf{y':x'} }$}
					child {
						node {$\mylabelC{\mathsf{x'} = 2}{\mathsf{y} \mapsto \mathsf{a}:\mathsf{y'} \\ \mathsf{x} \mapsto \mathsf{y':x'} }$}                           
							child {
								node {$\mylabelC{\varnothing}{\mathsf{y} \mapsto \mathsf{a}:\mathsf{y'} \\ \mathsf{x} \mapsto \mathsf{y':2} }$ \\ \blue{solved}}                           
								edge from parent node[right]{$\mathsf{Subst1}$}
							}
						edge from parent node[right]{$\mathsf{Orient1}$}
					}                           
				edge from parent node[right]{$\mathsf{Decomp2'}$}
			}
		edge from parent node[right]{$\mathsf{Subst3}$}
	}	 
;
\end{tikzpicture}
\caption{Unification example}
\label{Unifexample}
\end{figure}
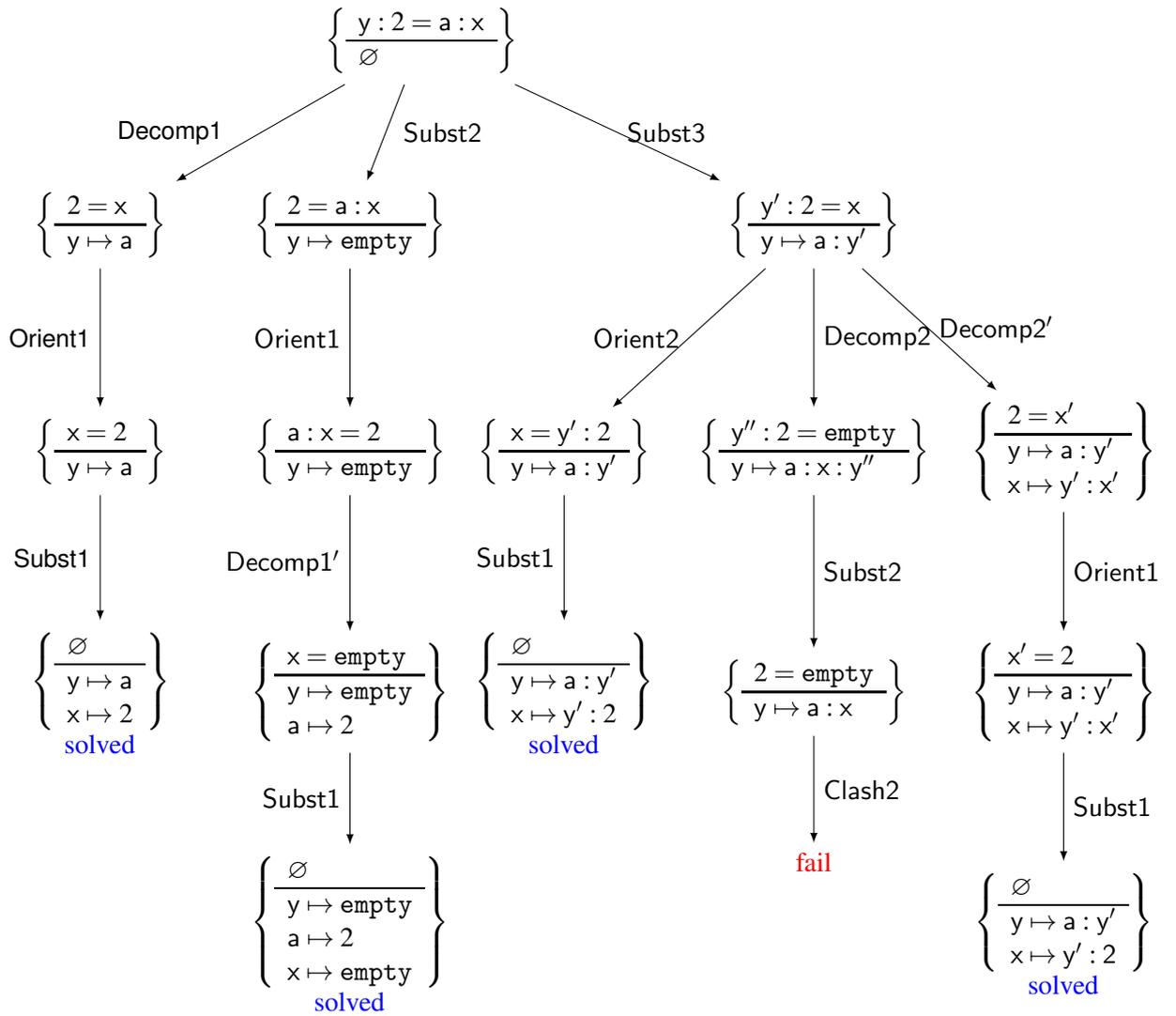

The algorithm is similar to the A-unification (word unification) algorithm presented in \cite{Schulz92a} which looks at the head of the problem equation. That algorithm terminates for the special case that the input problem has no repeated variables, and is sound and complete. Our approach can be seen as an extension from A-unification to AU-unification, to handle the unit equations, and presented in the rule-based style of \cite{Baader-Snyder01a}. In addition, our algorithm deals with GP's subtype system.

%% file: Termination.tex
\subsection{Termination and Soundness}
\label{subsec:unif-properties}

\setstretch{1}
We show that the unification algorithm terminates if the input problem contains no repeated list variables, where termination of the algorithm follows from termination of the relation $\Rightarrow$.

\begin{theorem}[Termination] If $P$ is a unification problem without repeated list variables, then there is no infinite sequence $P \Rightarrow P_1 \Rightarrow P_2 \Rightarrow \dots$
\end{theorem}
\begin{proof}
Define the \emph{size}\/ $|L|$ of an expression $L$\/ by
\begin{itemize}
\item $0$ if $L = \mtt{empty}$,
\item $1$ if $L$\/ is an expression of category Atom (see Figure \ref{fig:abstract_label_syntax}) or a list variable,
\item $|M|+|N|+1$ if $L = M:N$.
\end{itemize}
We define a lexicographic termination order by assigning to a unification problem $P = \langle L=M, \sigma \rangle$\/ the tuple $(n_1, n_2, n_3, n_4)$, where
\begin{itemize}
\item $n_1$ is the size of $P$, that is, $n_1 = |L| + |M|$;
\item $
n_2=
\begin{cases}
0 &\text{if~}L\text{ starts with a variable}\\
1 &\text{otherwise}
\end{cases}
$
\item $
n_3=
\begin{cases}
1 &\text{if~}\mrm{type}(\mtt{x}) > \mrm{type}(\mtt{y}) \\
0 &\text{otherwise}
\end{cases}
 $
 
where $\mtt{x}$ and $\mtt{y}$ are the starting symbols of $L$ and $M$
\item $n_4 = |L|$
\end{itemize}
The table in Figure \ref{ruleordering} shows that for each transformation step $P \dder P'$, the tuple associated with $P'$ is strictly smaller than the tuple associated with $P$\/ in the lexicographic order induced by the components $n_1$ to $n_4$.
\begin{figure}[!h]
\begin{center}
$\begin{array}{lcccccc}
 & n_1 & n_2 & n_3 & n_4 \\ \hline
\mathsf{Subst1-3} & >  \\
\mathsf{Decomp1-4} & >  \\
\mathsf{Remove} & >  \\
\mathsf{Orient1} & = & > \\
\mathsf{Orient4} & = & > \\
\mathsf{Orient3} & = & = & >\\
\mathsf{Orient2} & = & = & = & > \\
\end{array}$
\caption{Lexicographic termination order}
\label{ruleordering}
\end{center}
\end{figure}

For most rules, the table entries are easy to check. All the rules except for \textsf{Orient} decrease the size of the input equation $s \meqq t$. The \textsf{Orient} rules keep the size equal, but move variables and expressions of a bigger type to the left-hand side.
\end{proof}

%% file: CC.tex
%
%
%
\newpage
\begin{lemma}
\label{1.4}
If $P \Rightarrow P'$, then $\mathcal{U}(P) \supseteq \mathcal{U}(P')$\enspace.
\end{lemma}
\begin{proof}

We show that for each transformation rule, a unifier $\theta$ of the conclusion unifies the premise.

Note that for \textsf{Remove, Decomp3} and \textsf{Orient1-4}, we have that $\sigma_{P'} = \sigma_{P}$\enspace.

\begin{itemize}
\item \textsf{Remove}

$\begin{array}{lcl}
\theta \in \mathcal{U}(\langle \varnothing, \sigma_{P'} \rangle) & \iff & \theta \in \mathcal{U}(\langle \varnothing, \sigma_{P} \rangle)\\
& \iff & \theta \in \mathcal{U}(\langle \varnothing, \sigma_{P} \rangle) \wedge L\theta = L\theta \\
& \iff & \theta \in \mathcal{U}(\langle L=L, \sigma_{P} \rangle) \\
\end{array}$

\item \textsf{Decomp3}

$\begin{array}{lcl}
\theta \in \mathcal{U}(\langle L=M, \sigma_{P'} \rangle) & \iff & \theta \in \mathcal{U}(\langle L=M, \sigma_{P} \rangle)\\
& \iff & \theta \in \mathcal{U}(\langle L=M, \sigma_{P} \rangle) \wedge \mtt{s\theta = s\theta} \\
& \iff & \theta \in \mathcal{U}(\langle \mtt{s}:L=\mtt{s}:M, \sigma_{P} \rangle) \\
\end{array}$

\item \textsf{Decomp1} - We have $\mtt{x\theta = s\theta}$ and $L\theta=M\theta$. 

Then $(\mtt{x}:L)\theta = (\mtt{s}:L)\theta = (\mtt{s}:M)\theta$ as required.

\item \textsf{Decomp1'} - We have $\mtt{x\theta = a\theta}$ and $L\theta=M\theta$.

Then $(\mtt{x}:L)\theta = (\mtt{a}:L)\theta = (\mtt{a}:M)\theta$ as required.

\item \textsf{Decomp2} - We have $\mtt{x\theta =  (y:x')\theta}$ and $\mtt{(x'}:L)\theta= M\theta$.

Then $(\mtt{x}:L)\theta = (\mtt{y:x'}:L)\theta = (\mtt{y}:M)\theta$ as required.
\item \textsf{Decomp2'} - We have $\mtt{y\theta =  (x:y')\theta}$ and $L\theta= (\mtt{y'}:M)\theta$. 

Then $(\mtt{y}:M)\theta = (\mtt{x:y'}:M)\theta = (\mtt{x}:L)\theta$ as required.

\item \textsf{Decomp4} - We have $\mtt{x\theta =  a\theta}$ and $\mtt{y} \theta =  (\mtt{empty})\theta = \mtt{empty}$.

Then $(\mtt{a:y})\theta = (\mtt{x:y})\theta = (\mtt{x:empty})\theta = \mtt{x}\theta$ as required.

\item \textsf{Subst1} - We have $\mtt{x\theta} = L \theta$

$\begin{array}{lcl}
\theta \in \mathcal{U}(\langle \varnothing, \sigma_{S'} \rangle) & \iff & \theta \in \mathcal{U}(\langle \varnothing, \sigma_{S} \circ (\mtt{x} \mapsto L) \rangle)\\
& \iff & \theta \in \mathcal{U}(\langle \varnothing, \sigma_{S} \circ (\mtt{x} \mapsto L) \rangle) \wedge \mtt{x}\theta = L\theta\\
& \iff & \theta \in \mathcal{U}(\langle \mtt{x}=L, \sigma_{S} \rangle)\\
\end{array}$

\item \textsf{Subst2} - We have $ \mtt{x}\theta =  (\mtt{empty})\theta$ and $L\theta=M\theta$.

Then $(\mtt{x}:L)\theta = (\mtt{empty}:L)\theta = L\theta=M\theta$ as required.

\item \textsf{Subst3} - We have $\mtt{ x \theta=  (a:x')\theta}$ and $\mtt{(x'}:L)\theta=M\theta$.

Then $(\mtt{x}:L)\theta = (\mtt{a:x'}:L)\theta = (\mtt{a}:M)\theta$ as required.

\item \textsf{Orient1-4} - Since $=_{\mrm{AU}}$ is an equivalence relation and hence symmetric, a unifier of the conclusion is also a unifier of the premise.

\end{itemize}

\end{proof}

\begin{theorem}[Soundness]
\label{thm:soundness}
If $P \Rightarrow^{+} P'$ with $P' = \langle \varnothing, \sigma_{P'} \rangle$ in solved form, then $\sigma_{P'}$ is a unifier of $P$.
\end{theorem}
\begin{proof}
We have that $\sigma_{P'}$ is a unifier of $P'$ by definition. A simple induction with \autoref{1.4} shows that $\sigma_{P'}$ must be a unifier of $P$. 
\end{proof}

%% file: Completeness.tex
\subsection{Completeness}
In order to prove that our algorithm is complete, by \autoref{def:completeset} we have to show that for any unifier $\delta$, there is a unifier in our solution set that is more general.

Our proof involves using a \emph{selector} algorithm that takes a unification problem $\langle s \meqq t \rangle$ together with an arbitrary unifier $\delta$, and outputs a \emph{path} of the unification tree associated with a more general unifier than $\delta$. This is very similar to \cite{Siekmann78a} where completeness of a A-unification algorithm is shown via such selector.


\begin{lemma}[Selector Lemma]
\label{thm:prediction}
There exists an algorithm $\textsf{Select}(\delta, s \meqq t )$ that takes an equation $ s \meqq t $ and a unifier $\delta$ as input and produces a sequence of \emph{selections}  $B~=~( b_1, \ldots, b_k )$ such that:

\begin{itemize}
\item \textsf{Unify}($s \meqq t$) has a path specified by $B$\enspace.
\item For all selections $b \in B$: if $\sigma$ is the substitution corresponding to $b$, then there exists an instantiation $\lambda$ such that $\sigma \circ \lambda \leq \delta$\enspace.
\end{itemize}
\end{lemma}

For the selector algorithm and proof see \autoref{selector}.

For example, consider the unification problem $\langle \mtt{y:2=a:x} \rangle$ and $\delta = (\mtt{x \mapsto 1:2}, \mtt{y \mapsto a:1})$ as a unifier. The unification tree was shown in \autoref{Unifexample}. \textsf{Select}($\delta,\mtt{y:2=a:x}$) would produce selections $( \textsf{Subst3, Decomp2', Orient1, Subst1} )$, which corresponds to the right-most path in the tree. The unifier at the end of this path is $\sigma = (\mtt{x \mapsto y':2,y \mapsto a:y' })$ which is more general than~$\delta$ by instantiation $\lambda = (\mtt{y' \mapsto 1})$.

Now we are able to state our completeness result, which follows directly from Lemma~\ref{thm:prediction}.

\begin{theorem}[Completeness]
For every unifier $\delta$ of a unification problem $\langle s \meqq t\rangle$, there exist a unifier $\sigma$ generated by \textsf{Unify} such that $\sigma \leq \delta$\enspace.
\end{theorem}

%% file: completeness.tex
\section{Proving Completeness: The \textsf{Select} Algorithm} 
\label{selector}

As already stated, the \textsf{Select} algorithm takes a unification problem $s \meqq t$ together with a unifier $\delta$, and outputs a \emph{path} of the unification tree associated with a unifier that is at least as general $\delta$. 

The algorithm itself is presented below, together with several examples of how it works. Afterwards, we are able to prove the Selector Lemma meaning that \textsf{Select} is correct w.r.t. \textsf{Unify}.

\input{selector-preprocess} 
\input{predicates}
\input{selector}
\input{selector-examples}
\input{selector-lemma}

%% file: selector-preprocess.tex
\paragraph{Preprocessing}

Preprocessing of the problem $s=t$ w.r.t. a unifier $\delta$ involves expanding each variable to the correct number of symbols, i.e. $x \rightarrow x_1:\ldots:x_n$ where $n=|x\delta|$\ . 

Call $x$ the \emph{parent~symbol} of $x_1, \ldots, x_n$, and the expanded strings $\overline{s}$ and $\overline{t}$.

Furthermore, for each binding $\{x \mapsto empty\}$ in $\delta$, replace $x$ with $x_e$ in \={s} and \={t}. When \textsf{Select} sees such an empty list variable, it will be able to select Subst2.

Also, \emph{pad} the shorter string with extra $\mtt{empty}$ symbols to make \={s} and \={t} of equal length.

For example, the problem $a:x = y:b$ has unifier $\delta = (x \mapsto acb, y \mapsto acb )$ and the expansion would transform it into $a: x_1 : x_2 : x_3 = y_1: y_2 : y_3 : b$ where each "pseudo"-variable can hold exactly one symbol. 

The expanded \={s} and \={t} can be represented graphically as the diagram below:

\scalebox{0.8}{
\begin{tikzpicture}
\tikzstyle{every path}=[very thick]

\edef\sizetape{0.7cm}
\tikzstyle{tmtape}=[draw,minimum size=\sizetape]
\tikzstyle{tmhead}=[arrow box,draw,minimum size=.5cm,arrow box
arrows={east:.25cm, west:0.25cm}]

\begin{scope}
\begin{scope}[start chain=1 going right,node distance=-0.15mm]
    \node [on chain=1,tmtape](1) {a};
    \node [on chain=1,tmtape] {$x_1$};
    \node [on chain=1,tmtape] {$x_2$};
    \node [on chain=1,tmtape] {$x_3$};
    \node [on chain=1] {\textbf{\={s}}};
\end{scope}

\begin{scope}[yshift=-1.75cm,start chain=1 going right,node distance=-0.15mm]
    \node [on chain=1,tmtape](2) {$y_1$};
    \node [on chain=1,tmtape] {$y_2$};
    \node [on chain=1,tmtape] {$y_3$};
    \node [on chain=1,tmtape] (input) {b};
    \node [on chain=1] {\textbf{\={t}}};
\end{scope}

\node [draw=gray,dashed, fit= (1), inner sep=0.15cm,label=above:m] {};
\node [draw=gray,dashed, fit= (2), inner sep=0.15cm,label=above:n] {};

\end{scope}
\end{tikzpicture}}

%% file: predicates.tex
\paragraph{Predicates} ~\\
 
\begin{itemize} 
\item A basic look-ahead is implemented as a predicate $L(s,k)$ that returns the number of adjacent "pseudo"-variables of $s_k$ to the right of $s_k$. If $s_k$ has no "pseudo"-variables or a constant then $L(s,k)=0$.

For example, for the above strings:

\begin{tabular}{c|r}
$L(\overline{s},1) = 0$ & $a$ is a constant \\
$L(\overline{s},2) = 2$ & after $x_1$ there is $x_2$ and $x_3$  \\
$L(\overline{s},3) = 1$ & after $x_2$ there is $x_3$  \\
$L(\overline{s},4) = 0$ & after $x_3$ there is no $x_4$  \\ \\
$L(\overline{t},1) = 2$ & after $y_1$ there is $y_2$ and $y_3$  \\
$L(\overline{t},2) = 1$ & after $y_2$ there is $y_3$  \\
$L(\overline{t},3) = 0$ & after $y_3$ there is no $y_4$  \\ 
$L(\overline{t},4) = 0$ & $b$ is a constant \\
\end{tabular}

This predicate is really only useful for \emph{list} variables. 
 
\item Another predicate $$Type(<string>,<pos>) \to \{ListVar,~AtomVar,~AtomExpr,~EmptyListVar,~empty\} $$ returns the type of the symbol of $<string>$ at position $<pos>$.

For example for the string $\overline{s} = a:"a":x_e:y:empty$

\begin{tabular}{l}
$Type(\overline{s},1) = AtomVar$ \\
$Type(\overline{s},2) = AtomExpr$ \\
$Type(\overline{s},3) = EmptyListVar$ \\
$Type(\overline{s},4) = ListVar$ \\
$Type(\overline{s},5) = empty$ \\
\end{tabular}

\item Another predicate \[ LookAhead(<string>,<pos>) \to \{true, false\} \] determines if the word after the given symbol (or sequence of "pseudo"-variables when $Type(<string>,<pos>) = ListVar$) is not the $\mtt{empty}$ constant (or the blank symbol $\Delta$ at the end of the preprocessed string).

For example for the string $\overline{s} = x_1:x_2:a$

\begin{tabular}{lr}
$LookAhead(\overline{s},1) = true$ & because of a\\
$LookAhead(\overline{s},2) = true$ & because of a\\
$LookAhead(\overline{s},3) = false$ & because at end of string\\
\end{tabular}
\end{itemize}

%% file: selector.tex
\paragraph{The \textsf{Select} Algorithm} ~\\

The algorithm is essentially a 3-tape TM with look-ahead. The head looks at both input tapes simultaneously (the "head" of the current problem) as defined by pointers $m$ and $n$, and can move each pointer separately and only to the right.

It starts with preprocessing the input problem as explained above. Depending on the predicate values at $m$ and $n$, specific selections are recorded and pointers are updated. This is repeated until the pointers both reach the end of the input strings.

\textsf{Select} $(s,t,\delta) \to B$ :

Preprocess($s,t,\delta$) $\to \overline{s}, \overline{t}$  

$m = 1, n = 1$  // the head position

$B = ""$ // initial selection

\textbf{if} $m > length(\overline{s}) \text{~and~} n > length(\overline{t})$ \textbf{return} B

Case analysis on $L(\overline{s},m)$ and $L(\overline{t},n)$:

\begin{case} $L(\overline{s},m) < L(\overline{t},n)$ 

Since $L(\overline{t},n) > 0$, $\overline{t_n}$ must be a list variable

Do subcase analysis on type of $\overline{s_m}$ and look ahead results

\begin{subcase} $Type(\overline{s_m}) = ListVar$ ~~ /* $x:L=y:M$ */

/* It cannot be the case that the $LookAhead(\overline{s_m})$ returns $false$ (i.e. $L \neq empty$) as $\overline{t_n}$ is not of type $EmptyListVar$ ($L(\overline{t_n}) > 0$) */

~~~~~~select \textsf{Decomp2'} ~~ /* $x:L=y:M$ with $L \neq empty$ */

~~~~~~$m = m + 1 + L(\overline{s_m})$, $n = n + 1 + L(\overline{s_m})$ ~~ /* move head to end of shorter substring */
\end{subcase}

\begin{subcase} $Type(\overline{s_m}) = AtomVar$ ~~ /* $a:L=y:M$ */

~~~~~~ select \textsf{Orient3}, swap input tapes and pointers

\end{subcase}

\begin{subcase} $Type(\overline{s_m}) = AtomExpr$ ~~ /* $a:L=y:M$ */

~~~~~~ select \textsf{Orient1}, swap input tapes and pointers
\end{subcase}

\begin{subcase} $Type(\overline{s_m}) = EmptyListVar$ ~~ /* $x_e:L=y:M$ */

/* It cannot be the case that the $LookAhead(\overline{s_m})$ returns $false$ (i.e. $L \neq empty$) as $\overline{t_n}$ is not of type $EmptyListVar$ ($L(\overline{t_n}) > 0$) */

~~~~~~ select \textsf{Subst2}, $m$++, no change for $n$
\end{subcase}

\begin{subcase} $Type(\overline{s_m}) = empty$ 

/* cannot happen, as $\overline{t_n}$ is not of type $EmptyListVar$ ($L(\overline{t_n}) > 0$) and $\delta$ is a unifier */

\end{subcase} 
\end{case}

\begin{case} $L(\overline{s},m) > L(\overline{t},n)$

Since $L(\overline{s},m) > 0$, $\overline{s_m}$ must be a list variable

~~ \textbf{if} $LookAhead(\overline{s_m}) = false$ 

~~~~~~~~\textbf{then} \textbf{(Case 2.1)} 

~~~~~~~~~ select \textsf{Subst1} ~~ /* $\mtt{x} = \mtt{y}:M$ */

~~~~~~~~~ $m = m + 1 + L(\overline{s_m})$, $n = length(\overline{t}) + 1$ ~~ /* move head to end of both strings */

~~~~~~~~ \textbf{else} ~~ /* $L \neq empty$ */

\addtocounter{subcase}{1}

\begin{subcase} $Type(\overline{t_n}) = ListVar$

~~~~~~~~~select \textsf{Decomp2} ~~ /* $\mtt{x}:L=\mtt{y}:M$ with $L \neq M$ and $\mtt{x}$ should start with $\mtt{y}$ */

~~~~~~~~~ $m = m + 1 + L(\overline{t_n})$, $n = n + 1 + L(\overline{t_n})$ ~~ /* move head to end of shorter substring */
\end{subcase} 

\begin{subcase} $Type(\overline{t_n}) = AtomVar~or~AtomExpr$

~~~~~~~~~select \textsf{Subst3} ~~ /* $\mtt{x}:L=\mtt{a}:M$ with $L \neq M$ and $\mtt{x}$ should start with $\mtt{a}$ */

~~~~~~~~~ $m$++, $n$++
\end{subcase} 

\begin{subcase} $Type(\overline{t_n}) = EmptyListVar$

~~~~~~~~~select \textsf{Decomp2} ~~ /* $\mtt{x}:L=y_e:M$ with $L \neq M$ and $\mtt{x}$ should start with $y$ */

~~~~~~~~~ $m$++, $n$++
\end{subcase} 
 
\begin{subcase} $Type(\overline{t_n}) = empty$ 

/* cannot happen, same reason as 1.5 */
\end{subcase}

\end{case}

\begin{case}  $L(\overline{s},m) = L(\overline{t},n)$

\begin{subcase} $L(\overline{s},m) = L(\overline{t},n) > 0$

Both $\overline{s}_m$ and $\overline{t}_n$ must be list variables

~~ \textbf{if} $LookAhead(\overline{s_m}) = false$ 

~~~~~~~~~ select \textsf{Subst1} ~~ /* $\mtt{x} = \mtt{y}:M$ */

~~~~~~~~~ $m = m + 1 + L(\overline{s_m})$, $n = length(\overline{t}) + 1$ ~~ /* move head to end of both strings */

~~ \textbf{else} ~~ /* $L \neq empty$ */

~~~~~~~~~ select \textsf{Decomp1} ~~ /* $\mtt{x}:L = \mtt{y}:M$ */

~~~~~~~~~ $m = m + 1 + L(\overline{s_m})$, $n = n + 1 + L(\overline{s_m})$ ~/* move head to next pair of unrelated symbols */

\end{subcase}

\begin{subcase} $L(\overline{s},m) = L(\overline{t},n) = 0$ 

/* This is the largest subcase because $\overline{s_m}$ and $\overline{t_n}$ can be of any of the possible types. */

Do case analysis on $Type(\overline{s}_m)$ and $Type(\overline{t}_n)$:

\textbf{3.2.1. } $(ListVar,ListVar)$ -- same as 3.1

\textbf{3.2.2. } $(ListVar,AtomVar)$ -- same as 3.1

\textbf{3.2.3. } $(ListVar,AtomExpr)$ -- same as 3.1

\textbf{3.2.4. } $(ListVar,EmptyListVar)$ -- same as 2.4

\textbf{3.2.5. } $(ListVar,empty)$ -- cannot happen as $\delta$ is a unifier

\textbf{3.2.6. } $(AtomVar,ListVar)$ -- Orient3 (like 1.2)

\textbf{3.2.7-8 } $(AtomVar,AtomVar)$  or  $(AtomVar,AtomExpr)$ 

~~ \textbf{if} $LookAhead(\overline{s_m}) = false$ 

~~~~~~ \textbf{if} $LookAhead(\overline{t_n}) = false$ 

~~~~~~~~~~~~ select \textsf{Subst1} ~~ /* $\mtt{a} = \mtt{b}$  */

~~~~~~~~~~~~ $m$++, $n$++ ~~ /* move head to end of both strings */

~~~~~~ \textbf{else} 

~~~~~~~~~~~~ select \textsf{Decomp4} ~~ /* $\mtt{a} = \mtt{b:y}$ with $\delta(\mtt{y}) = \mtt{empty}$ */

~~~~~~~~~~~~ $m$++, $n$++ ~~ /* move head to next pair of unrelated symbols */

~~ \textbf{else} ~~ /* $L \neq empty$ */

~~~~~~~~~ select \textsf{Decomp1'} ~~ /* $\mtt{a}:L = \mtt{b}:M$ */

~~~~~~~~~ $m$++, $n$++ ~~ /* move head to next pair of unrelated symbols */

\textbf{3.2.9. } $(AtomVar,EmptyListVar)$ -- same as 1.2

\textbf{3.2.10. } $(AtomVar,empty)$ /* cannot happen due to failure lemmata */

\textbf{3.2.11. } $(AtomExpr,ListVar)$ -- same as 1.3
 
\textbf{3.2.12. } $(AtomExpr,AtomVar)$ -- same as 1.3 

\textbf{3.2.13. } $(AtomExpr,AtomExpr)$ -- 

/* must be the same expression due to not considering the subunification algorithm for String-Char */

~~~~~ select \textsf{Decomp3}, $m$++, $n$++
 
\textbf{3.2.14. } $(AtomExpr,EmptyListVar)$ -- same as 1.3
 
\textbf{3.2.15. } $(AtomExpr,empty)$ -- /* cannot happen due to failure lemmata  */

\textbf{3.2.16-18. } $(EmptyListVar,ListVar)$ or $(EmptyListVar,AtomVar)$ or $(EmptyListVar,AtomExpr)$ 

/* It must be the case that there is something following $\overline{s_m}$ for $\delta$ to be a unifier. */

~~~~~ select \text{Subst2}, $m$++ (like 1.4)

\textbf{3.2.19. } $(EmptyListVar,EmptyListVar)$ -- same as 3.1

\textbf{3.2.20. } $(EmptyListVar,empty)$ -- select \textsf{Subst1}, $m$++, $n = length(\overline{t_n})+1$
 
\textbf{3.2.21. } $(empty,ListVar)$ /* cannot happen as $\delta$ is a unifier */
 
\textbf{3.2.22 \& 3.2.23 } $(empty,AtomVar)$ or $(empty,AtomExpr)$ /* cannot happen due to failure lemmata  */

\textbf{3.2.24. } $(empty,EmptyListVar)$ -- select \textsf{Orient4}, swap tapes and pointers
 
\textbf{3.2.25. } $(empty,empty)$ select \textsf{Remove} /* nothing left to solve */

\end{subcase}
\end{case}

%% file: selector-examples.tex
\paragraph{Examples}

Below are some examples of how the algorithm behaves when presented with a unification problem and a unifier.

The first two examples are of the problem $P = \langle \mtt{y:2 \meqq a:x} \rangle$, with its unification tree already shown in \autoref{Unifexample}. The third example is of the problem $P = \langle \mtt{n \meqq y:n} \rangle$ that has a single most general unifier $\sigma = \{ \mtt{y \mapsto empty} \}$.

\input{example1}
\input{example2}
\input{example3}

%% file: example1.tex
\begin{example}
 \[P = \langle \mtt{y:2 \meqq a:x} \rangle \] 
 \[\delta = \{ \mtt{a \mapsto 2,~x \mapsto empty,~y \mapsto empty} \} \]
 
\scalebox{0.95}{ 
 
\begin{tikzpicture}
\tikzstyle{every path}=[very thick]

\edef\sizetape{0.7cm}
\tikzstyle{tmtape}=[draw,minimum size=\sizetape]
\tikzstyle{tmhead}=[arrow box,draw,minimum size=.5cm,arrow box
arrows={east:.25cm, west:0.25cm}]

\begin{scope}
\begin{scope}[start chain=1 going right,node distance=-0.15mm]
    \node [on chain=1,tmtape](1) {$y_e$};
    \node [on chain=1,tmtape] {2};
    \node [on chain=1,tmtape] {$\Delta$};
    \node [on chain=1](3) {\textbf{\={s}}};
\end{scope}

\begin{scope}[yshift=-1.75cm,start chain=1 going right,node distance=-0.15mm]
    \node [on chain=1,tmtape](2) {$a$};
    \node [on chain=1,tmtape] {$x_e$};
    \node [on chain=1,tmtape] (input) {$\Delta$};
    \node [on chain=1](4) {\textbf{\={t}}};
\end{scope}

\begin{scope}[yshift=-3.0cm]
	\node[xshift=0.7cm](s) {select \textsf{Subst2}};
	\node[below of=s,yshift=0.5cm] {$m$++};
	\node[below of=s] {(case 3.2.17)};
\end{scope}

\node [draw=gray,dashed, fit= (1), inner sep=0.13cm,label=above:{$m=1$}] {};
\node [draw=gray,dashed, fit= (2), inner sep=0.13cm,label=above:{$n=1$}] {};
\node at ($(3)!0.5!(4)$) {$\rightarrow$};

\end{scope}

\begin{scope}[xshift=3cm]
\begin{scope}[start chain=1 going right,node distance=-0.15mm]
    \node [on chain=1,tmtape] {$y_e$};
    \node [on chain=1,tmtape](1) {2};
    \node [on chain=1,tmtape] {$\Delta$};
    \node [on chain=1](3) {\textbf{\={s}}};
\end{scope}

\begin{scope}[yshift=-1.75cm,start chain=1 going right,node distance=-0.15mm]
    \node [on chain=1,tmtape](2) {$a$};
    \node [on chain=1,tmtape] {$x_e$};
    \node [on chain=1,tmtape] (input) {$\Delta$};
    \node [on chain=1](4) {\textbf{\={t}}};
\end{scope}

\begin{scope}[yshift=-3.0cm]
	\node[xshift=0.7cm](s) {select \textsf{Orient1}};
	\node[below of=s,yshift=0.5cm] {swap tapes \& pointers};
	\node[below of=s] {(case 3.2.12)};
\end{scope}

\node [draw=gray,dashed, fit= (1), inner sep=0.13cm,label=above:{$m=2$}] {};
\node [draw=gray,dashed, fit= (2), inner sep=0.13cm,label=above:{$n=1$}] {};
\node at ($(3)!0.5!(4)$) {$\rightarrow$};

\end{scope}

\begin{scope}[xshift=6cm]
\begin{scope}[yshift=-1.75cm,start chain=1 going right,node distance=-0.15mm]
    \node [on chain=1,tmtape] {$y_e$};
    \node [on chain=1,tmtape](1) {2};
    \node [on chain=1,tmtape] {$\Delta$};
    \node [on chain=1](4) {\textbf{\={t}}};
\end{scope}

\begin{scope}[start chain=1 going right,node distance=-0.15mm]
    \node [on chain=1,tmtape](2) {$a$};
    \node [on chain=1,tmtape] {$x_e$};
    \node [on chain=1,tmtape] (input) {$\Delta$};
    \node [on chain=1](3) {\textbf{\={s}}};
\end{scope}

\begin{scope}[yshift=-3.0cm]
	\node(s)[xshift=0.7cm] {select \textsf{Decomp1'}};
	\node[below of=s,yshift=0.5cm] {$m$++,$n$++};
	\node[below of=s] {(case 3.2.8)};
\end{scope}

\node [draw=gray,dashed, fit= (1), inner sep=0.13cm,label=above:{$n=2$}] {};
\node [draw=gray,dashed, fit= (2), inner sep=0.13cm,label=above:{$m=1$}] {};
\node at ($(3)!0.5!(4)$) {$\rightarrow$};

\end{scope}

\begin{scope}[xshift=9cm]
\begin{scope}[yshift=-1.75cm,start chain=1 going right,node distance=-0.15mm]
    \node [on chain=1,tmtape] {$y_e$};
    \node [on chain=1,tmtape] {2};
    \node [on chain=1,tmtape](1) {$\Delta$};
    \node [on chain=1](4) {\textbf{\={t}}};
\end{scope}

\begin{scope}[start chain=1 going right,node distance=-0.15mm]
    \node [on chain=1,tmtape] {$a$};
    \node [on chain=1,tmtape](2) {$x_e$};
    \node [on chain=1,tmtape] (input) {$\Delta$};
    \node [on chain=1](3) {\textbf{\={s}}};
\end{scope}

\begin{scope}[yshift=-3.0cm]
	\node(s)[xshift=0.7cm] {select \textsf{Subst1} };
	\node[below of=s,yshift=0.5cm] {$m$++, $n=3$};
	\node[below of=s] {(case 3.2.20)};
\end{scope}

\node [draw=gray,dashed, fit= (2), inner sep=0.13cm,label=above:{$m=2$}] {};
\node [draw=gray,dashed, fit= (1), inner sep=0.13cm,label=above:{$n=3$}] {};
\node[xshift=0.3cm] at ($(3)!0.5!(4)$) {$\rightarrow$};

\end{scope}

\begin{scope}[xshift=12cm]
\begin{scope}[start chain=1 going right,node distance=-0.15mm]
    \node [on chain=1,tmtape] {$y_e$};
    \node [on chain=1,tmtape] {2};
    \node [on chain=1,tmtape](1) {$\Delta$};
    \node [on chain=1] {\textbf{\={s}}};
\end{scope}

\begin{scope}[yshift=-1.75cm,start chain=1 going right,node distance=-0.15mm]
    \node [on chain=1,tmtape] {$a$};
    \node [on chain=1,tmtape] {$x_e$};
    \node [on chain=1,tmtape](2) {$\Delta$};
    \node [on chain=1] {\textbf{\={t}}};
\end{scope}

\begin{scope}[yshift=-3.25cm]
	\node(s)[xshift=0.7cm] {end };
\end{scope}

\node [draw=gray,dashed, fit= (1), inner sep=0.13cm,label=above:{$m=3$}] {};
\node [draw=gray,dashed, fit= (2), inner sep=0.13cm,label=above:{$n=3$}] {};
\end{scope}

\end{tikzpicture}}
\end{example}

%% file: example2.tex
\begin{example}
 \[P = \langle \mtt{y:2 \meqq a:x} \rangle \] 
 \[\delta = \{ \mtt{x \mapsto a:1,~y \mapsto 1:2} \} \] 

\scalebox{0.95}{
\begin{tikzpicture}
\tikzstyle{every path}=[very thick]

\edef\sizetape{0.7cm}
\tikzstyle{tmtape}=[draw,minimum size=\sizetape]
\tikzstyle{tmhead}=[arrow box,draw,minimum size=.5cm,arrow box
arrows={east:.25cm, west:0.25cm}]

\begin{scope}
\begin{scope}[start chain=1 going right,node distance=-0.15mm]
    \node [on chain=1,tmtape](1) {$y_1$};
    \node [on chain=1,tmtape] {$y_2$};
    \node [on chain=1,tmtape] {2};
    \node [on chain=1,tmtape] {$\Delta$};
    \node [on chain=1](3) {\textbf{\={s}}};
\end{scope}

\begin{scope}[yshift=-1.75cm,start chain=1 going right,node distance=-0.15mm]
    \node [on chain=1,tmtape](2) {$a$};
    \node [on chain=1,tmtape] {$x_1$};
    \node [on chain=1,tmtape] {$x_2$};
    \node [on chain=1,tmtape] (input) {$\Delta$};
    \node [on chain=1](4) {\textbf{\={t}}};
\end{scope}

\begin{scope}[yshift=-3.0cm]
	\node[xshift=0.7cm](s) {select \textsf{Subst3}};
	\node[below of=s,yshift=0.5cm] {$m++,n++$};
	\node[below of=s] {(case 2.3)};
\end{scope}

\node [draw=gray,dashed, fit= (1), inner sep=0.13cm,label=above:{$m=1$}] {};
\node [draw=gray,dashed, fit= (2), inner sep=0.13cm,label=above:{$n=1$}] {};
\node[xshift=0.5cm] at ($(3)!0.5!(4)$) {$\rightarrow$};

\end{scope}

\begin{scope}[xshift=4.5cm]
\begin{scope}[start chain=1 going right,node distance=-0.15mm]
    \node [on chain=1,tmtape] {$y_1$};
    \node [on chain=1,tmtape](1) {$y_2$};
    \node [on chain=1,tmtape] {2};
    \node [on chain=1,tmtape] {$\Delta$};
    \node [on chain=1](3) {\textbf{\={s}}};
\end{scope}

\begin{scope}[yshift=-1.75cm,start chain=1 going right,node distance=-0.15mm]
    \node [on chain=1,tmtape] {$a$};
    \node [on chain=1,tmtape](2) {$x_1$};
    \node [on chain=1,tmtape] {$x_2$};
    \node [on chain=1,tmtape] (input) {$\Delta$};
    \node [on chain=1](4) {\textbf{\={t}}};
\end{scope}

\begin{scope}[yshift=-3.0cm]
	\node[xshift=0.7cm](s) {select \textsf{Decomp2'}};
	\node[below of=s,yshift=0.5cm] {$m++,n++$};
	\node[below of=s] {(case 2.1)};
\end{scope}

\node [draw=gray,dashed, fit= (1), inner sep=0.13cm,label=above:{$m=2$}] {};
\node [draw=gray,dashed, fit= (2), inner sep=0.13cm,label=above:{$n=2$}] {};
\node[xshift=0.5cm] at ($(3)!0.5!(4)$) {$\rightarrow$};

\end{scope}

\begin{scope}[xshift=9cm]
\begin{scope}[start chain=1 going right,node distance=-0.15mm]
    \node [on chain=1,tmtape] {$y_1$};
    \node [on chain=1,tmtape] {$y_2$};
    \node [on chain=1,tmtape](1) {2};
    \node [on chain=1,tmtape] {$\Delta$};
    \node [on chain=1](3) {\textbf{\={s}}};
\end{scope}

\begin{scope}[yshift=-1.75cm,start chain=1 going right,node distance=-0.15mm]
    \node [on chain=1,tmtape] {$a$};
    \node [on chain=1,tmtape] {$x_1$};
    \node [on chain=1,tmtape](2) {$x_2$};
    \node [on chain=1,tmtape] (input) {$\Delta$};
    \node [on chain=1](4) {\textbf{\={t}}};
\end{scope}

\begin{scope}[yshift=-3.0cm]
	\node(s)[xshift=0.7cm] {select \textsf{Orient1}};
	\node[below of=s,yshift=0.5cm] {swap tapes \& pointers};
	\node[below of=s] {(case 3.2.11)};
\end{scope}

\node [draw=gray,dashed, fit= (2), inner sep=0.13cm,label=above:{$m=3$}] {};
\node [draw=gray,dashed, fit= (1), inner sep=0.13cm,label=above:{$n=3$}] {};
\node[xshift=0.5cm] at ($(3)!0.5!(4)$) {$\rightarrow$};

\end{scope}

\begin{scope}[xshift=3cm,yshift=-5.5cm]
\begin{scope}[yshift=-1.75cm,start chain=1 going right,node distance=-0.15mm]
    \node [on chain=1,tmtape] {$y_1$};
    \node [on chain=1,tmtape] {$y_2$};
    \node [on chain=1,tmtape](1) {2};
    \node [on chain=1,tmtape] {$\Delta$};
    \node [on chain=1](3) {\textbf{\={t}}};
\end{scope}

\begin{scope}[start chain=1 going right,node distance=-0.15mm]
    \node [on chain=1,tmtape] {$a$};
    \node [on chain=1,tmtape] {$x_1$};
    \node [on chain=1,tmtape](2) {$x_2$};
    \node [on chain=1,tmtape] {$\Delta$};
    \node [on chain=1](4) {\textbf{\={s}}};
\end{scope}

\begin{scope}[yshift=-3.0cm]
	\node(s)[xshift=0.7cm] {select \textsf{Subst1} };
	\node[below of=s,yshift=0.5cm] {$m++, n++$};
	\node[below of=s] {(case 3.2.3)};
\end{scope}

\node [draw=gray,dashed, fit= (2), inner sep=0.13cm,label=above:{$m=3$}] {};
\node [draw=gray,dashed, fit= (1), inner sep=0.13cm,label=above:{$n=3$}] {};
\node[xshift=1cm] at ($(3)!0.5!(4)$) {$\rightarrow$};
\node[xshift=-4cm] at ($(3)!0.5!(4)$) {$\rightarrow$};

\end{scope}

\begin{scope}[xshift=9cm,yshift=-5.5cm]
\begin{scope}[start chain=1 going right,node distance=-0.15mm]
    \node [on chain=1,tmtape] {$y_1$};
    \node [on chain=1,tmtape] {$y_2$};
    \node [on chain=1,tmtape] {2};
    \node [on chain=1,tmtape](1) {$\Delta$};
    \node [on chain=1] {\textbf{\={s}}};
\end{scope}

\begin{scope}[yshift=-1.75cm,start chain=1 going right,node distance=-0.15mm]
    \node [on chain=1,tmtape] {$a$};
    \node [on chain=1,tmtape] {$x_1$};
    \node [on chain=1,tmtape] {$x_2$};
    \node [on chain=1,tmtape](2) {$\Delta$};
    \node [on chain=1] {\textbf{\={t}}};
\end{scope}

\begin{scope}[yshift=-3.25cm]
	\node(s)[xshift=0.7cm] {end };
\end{scope}

\node [draw=gray,dashed, fit= (1), inner sep=0.13cm,label=above:{$m=4$}] {};
\node [draw=gray,dashed, fit= (2), inner sep=0.13cm,label=above:{$n=4$}] {};
\end{scope}

\end{tikzpicture}}
\end{example}

%% file: example3.tex
\begin{example}
 \[P = \langle \mtt{n \meqq y:n} \rangle \\ 
 \delta = \{ \mtt{y \mapsto empty} \} \]
 
\scalebox{0.95}{
\begin{tikzpicture}
\tikzstyle{every path}=[very thick]

\edef\sizetape{0.7cm}
\tikzstyle{tmtape}=[draw,minimum size=\sizetape]
\tikzstyle{tmhead}=[arrow box,draw,minimum size=.5cm,arrow box
arrows={east:.25cm, west:0.25cm}]

\begin{scope}
\begin{scope}[start chain=1 going right,node distance=-0.15mm]
    \node [on chain=1,tmtape](1) {$n$};
    \node [on chain=1,tmtape] {$\mtt{e}$};
    \node [on chain=1,tmtape] {$\Delta$};
    \node [on chain=1](3) {\textbf{\={s}}};
\end{scope}

\begin{scope}[yshift=-1.75cm,start chain=1 going right,node distance=-0.15mm]
    \node [on chain=1,tmtape](2) {$y_e$};
    \node [on chain=1,tmtape] {$n$};
    \node [on chain=1,tmtape] (input) {$\Delta$};
    \node [on chain=1](4) {\textbf{\={t}}};
\end{scope}

\begin{scope}[yshift=-3.0cm]
	\node[xshift=0.7cm](s) {select \textsf{Orient3}};
	\node[below of=s,yshift=0.5cm] {swap tapes \& pointers};
	\node[below of=s] {(case 3.2.9)};
\end{scope}

\node [draw=gray,dashed, fit= (1), inner sep=0.13cm,label=above:{$m=1$}] {};
\node [draw=gray,dashed, fit= (2), inner sep=0.13cm,label=above:{$n=1$}] {};
\node[xshift=0.1cm] at ($(3)!0.5!(4)$) {$\rightarrow$};

\end{scope}

\begin{scope}[xshift=3cm]
\begin{scope}[yshift=-1.75cm,start chain=1 going right,node distance=-0.15mm]
    \node [on chain=1,tmtape](1) {$n$};
    \node [on chain=1,tmtape] {$\mtt{e}$};
    \node [on chain=1,tmtape] {$\Delta$};
    \node [on chain=1](3) {\textbf{\={t}}};
\end{scope}

\begin{scope}[start chain=1 going right,node distance=-0.15mm]
    \node [on chain=1,tmtape](2) {$y_e$};
    \node [on chain=1,tmtape] {$n$};
    \node [on chain=1,tmtape] (input) {$\Delta$};
    \node [on chain=1](4) {\textbf{\={s}}};
\end{scope}

\begin{scope}[yshift=-3.0cm]
	\node[xshift=0.7cm](s) {select \textsf{Subst2}};
	\node[below of=s,yshift=0.5cm] {$m$++};
	\node[below of=s] {(case 3.2.17)};
\end{scope}

\node [draw=gray,dashed, fit= (2), inner sep=0.13cm,label=above:{$m=1$}] {};
\node [draw=gray,dashed, fit= (1), inner sep=0.13cm,label=above:{$n=1$}] {};
\node[xshift=0.1cm] at ($(3)!0.5!(4)$) {$\rightarrow$};

\end{scope}

\begin{scope}[xshift=6cm]
\begin{scope}[yshift=-1.75cm,start chain=1 going right,node distance=-0.15mm]
    \node [on chain=1,tmtape](1) {$n$};
    \node [on chain=1,tmtape] {$\mtt{e}$};
    \node [on chain=1,tmtape] {$\Delta$};
    \node [on chain=1](3) {\textbf{\={t}}};
\end{scope}

\begin{scope}[start chain=1 going right,node distance=-0.15mm]
    \node [on chain=1,tmtape] {$y_e$};
    \node [on chain=1,tmtape](2) {$n$};
    \node [on chain=1,tmtape] (input) {$\Delta$};
    \node [on chain=1](4) {\textbf{\={s}}};
\end{scope}

\begin{scope}[yshift=-3.0cm]
	\node[xshift=0.7cm](s) {select \textsf{Subst1}};
	\node[below of=s,yshift=0.5cm] {$m$++,$n$++};
	\node[below of=s] {(case 3.2.7)};
\end{scope}

\node [draw=gray,dashed, fit= (2), inner sep=0.13cm,label=above:{$m=2$}] {};
\node [draw=gray,dashed, fit= (1), inner sep=0.13cm,label=above:{$n=1$}] {};
\node[xshift=0.1cm] at ($(3)!0.5!(4)$) {$\rightarrow$};
\end{scope}

\begin{scope}[xshift=9cm]
\begin{scope}[yshift=-1.75cm,start chain=1 going right,node distance=-0.15mm]
    \node [on chain=1,tmtape] {$n$};
    \node [on chain=1,tmtape](1) {$\mtt{e}$};
    \node [on chain=1,tmtape] {$\Delta$};
    \node [on chain=1](3) {\textbf{\={t}}};
\end{scope}

\begin{scope}[start chain=1 going right,node distance=-0.15mm]
    \node [on chain=1,tmtape] {$y_e$};
    \node [on chain=1,tmtape] {$n$};
    \node [on chain=1,tmtape](2) {$\Delta$};
    \node [on chain=1](4) {\textbf{\={s}}};
\end{scope}

\begin{scope}[yshift=-3.25cm]
	\node(s)[xshift=0.7cm] {end };
\end{scope}

\node [draw=gray,dashed, fit= (1), inner sep=0.13cm,label=above:{$n=2$}] {};
\node [draw=gray,dashed, fit= (2), inner sep=0.13cm,label=above:{$m=3$}] {};
\end{scope}

\end{tikzpicture}}
\end{example}

%% file: selector-lemma.tex
\setcounter{lemma}{4}
\begin{lemma}[Selector Lemma]
Let $B$ be the sequence of selections defined by $\textsf{Select}(\delta, s \meqq t )$ where $\delta(s)=\delta(t)$

1) \textsf{Unify}($s \meqq t$) has a path specified by $B$\enspace.

2) For all selections $b \in B$:

\begin{itemize}
\item 2.1) the symbols examined by \textsf{Unify} at $b$ are parent symbols of the symbols \textsf{Select} sees at head position $(m,n)$ for selection $b$
\item 2.2) if $\sigma$ is the substitution corresponding to selection $b$, then 
$ \sigma \leq \delta$, i.e. exist substitution $\lambda$ such that $\sigma \circ \lambda \subseteq \delta $
\end{itemize}
\end{lemma}

\begin{proof}
By induction on the number of the selections $(b_1,b_2,\ldots)$ of $B$

\paragraph{Base Case} The initial selection $b_1$ of \textsf{Select} is based on the first symbols of $\overline{s}$ and $\overline{t}$ as this is the initial head position ($m$ and $n$ are initialized to 1). The transformation rules of \textsf{Unify} also examine the head of the initial problem. Also, \textsf{Select} does not throw away list variables that should be removed due to Subst2, so condition 2.1) is trivially satisfied. For 2.2), $\sigma_1 = \varnothing$ by definition of initial unification problem, and let $\lambda_1 = \varnothing$. Therefore $\sigma_1 \circ \lambda_1 = \varnothing \subseteq \delta$ as required.

\paragraph{Hypothesis} Assume the statements of the theorem are true for selection $b_k$ with a corresponding node in the unification tree $\langle P_k, \sigma_k \rangle$ and let $(m,n)$ be the head position of \textsf{Select} before outputting selection $b_{k+1}$.

2.1) the symbols at the head of the problem $P_k$  ($\mtt{x}$ and $\mtt{y}$) are parent symbols of  $\overline{s}_m , \overline{t}_n$ 

2.1) exists substitution  $\lambda_k$ such that  $\sigma_k \circ \lambda _k \subseteq \delta $

\paragraph{Inductive case} There are several possible cases for the selection $b_{n+1}$. We have to show that from $b_n$,

\begin{enumerate}[i]
\item \label{en:i} \textsf{Unify} \emph{can} generate selection $b_{k+1}$ (i.e. the selected rule is applicable) and
\item \label{en:ii}
The symbols at the head of the next unification problem $P_{k+1}$ (where $P_k \Rightarrow P_{k+1}$ via the rule $b_{k+1}$) are parent symbols of $\overline{s}_{m_{new}}$ and $\overline{t}_{n_{new}}$
\item \label{en:iii} exists instantiation $\lambda_{n+1}$ such that $\sigma_{n+1}~\circ~\lambda_{n+1}~\subseteq~\delta$
\end{enumerate}

For all \textsf{Orient} rules, we can notice that:
\begin{itemize}
\item they do not change the head of the unification problem nor they generate a supplementary substitution $\sigma_{k+1}$, so conditions \ref{en:ii} and \ref{en:iii} trivially hold
\item they always swap left-hand side with right-hand side; so does \textsf{Select} for each \textsf{Orient} selection
\end{itemize}
Therefore, we only need to prove condition \ref{en:i} for each \textsf{Orient} selection.

We now have to consider all the cases from $P_k$ to $P_{k+1}$ in \textsf{Select} via rule $b_{k+1}$.

\setcounter{case}{0}
\begin{case} $L(\overline{s},m) < L(\overline{t},n)$ 

It must be the case that $\overline{t}_n$ is a list variable because $L(\overline{t},n) > 0$ only for list variables.

Call $\mtt{x}$ and $\mtt{y}$ the parent symbols of $\overline{s_m}$ and $\overline{t_n}$  in $s$ and $t$.

\begin{subcase} $Type(\overline{s_m}) = ListVar$

 By hypothesis, it follows that $\mtt{x}$ and $\mtt{y}$ are also list variables. 

Suppose nothing follows $\mtt{x}$ (i.e. $LookAhead(\overline{s}_m) = false$). Then the unification problem at node $b_k$ must be of the form $\mtt{x} = \mtt{y}:M$. However, we have that  $L(\overline{s},m) < L(\overline{t},n)$ which contradicts that $\delta$ is a unifier.

Therefore, there must be something following $\mtt{x}$. Then the current unification problem must be of the form $\mtt{x}:L = \mtt{y}:M$. Then rule \textsf{Decomp2'} is applicable (condition \ref{en:i} satisfied). Because $\mtt{x}$ has the shorter expansion by $\delta$, $\overline{t}_{n+L(\overline{s}_m)+1}$ is a child variable of $\mtt{y}$ and condition \ref{en:ii} holds. 

For \ref{en:iii}, let $\lambda_k' = (\overline{s_n} \ldots \overline{s_{m+L(\overline{s}_m)}} \mapsto T)$ where $T = \delta(\overline{s_m} \ldots \overline{s_{m+L(\overline{s}_m)}})$. For $\delta$ to be a unifier, it must contain term(s) $(\overline{t_n} \ldots \overline{t_{n+L(\overline{s}_m)}} \mapsto T)$:

\begin{tabular}{lcl}
$\delta$ &$\supseteq$& $\sigma_k \circ \lambda_k \cup (\overline{t_n} \ldots \overline{t_{n+L(\overline{s}_m)}} \mapsto T)$ \\
&=& $\sigma_k \circ \lambda_k \cup ((\overline{t_n} \ldots \overline{t_{n+L(\overline{s}_m)}} \mapsto \overline{s_m} \ldots \overline{s_{m+L(\overline{s}_m)}}) \circ \lambda_k')$ (same as saying $\mtt{y}$ starts with $\mtt{x}$) \\
&=& $\sigma_k \circ \lambda_k \cup ((\mtt{y \mapsto x:y'}) \circ \lambda_k')$ \\
&=& $\sigma_k \circ (\mtt{y \mapsto x:y'}) \circ \lambda_k \circ \lambda_k'$ (since no repeated list variables)\\ 
&=& $\sigma_{k+1} \circ \lambda_{k+1}$ as required (set $\lambda_{k+1} = \lambda_k \circ \lambda_k'$)\\
\end{tabular}

\end{subcase}

\begin{subcase} $Type(\overline{s_m}) = AtomVar$

It holds by hypothesis that $\mtt{x}$ and $\mtt{y}$ are atom and list variables so rule \textsf{Orient3} is applicable.
\end{subcase}

\begin{subcase} $Type(\overline{s_m}) = AtomExpr$ 

The argument is the same as above except that $\mtt{x}$ is an atom expression and \textsf{Orient1} becomes applicable.
\end{subcase}

\begin{subcase} $Type(\overline{s_m}) = EmptyListVar$

Because $L(\overline{t}_n) > L(\overline{s}_m) = 0$, there must be something after $\overline{s_{m+L(\overline{s}_m)}}$ (i.e. $LookAhead(\overline{s}_m) = true$) for $\delta$ to be a unifier. Also, by hypothesis $\mtt{x}$ must be a list variable since $\mtt{x}\delta~=~\mtt{empty}$. Therefore rule \textsf{Subst2} is applicable (condition \ref{en:i} holds). The resulting unification problem is $(\{L = M\} \{\mtt{x} \mapsto \texttt{empty}\} ;~\sigma_k \circ \{\mtt{x} \mapsto \texttt{empty}\})$. \textsf{Select} moves the position $m$ by 1 and keeps $n$, therefore condition \ref{en:ii} also holds. 

For condition \ref{en:iii}:

\begin{tabular}{lcl}
$\delta$ &$\supseteq$& $\sigma_k \circ \lambda_k \cup (\mtt{x} \mapsto \mtt{empty})$ \\
&=& $\sigma_k \circ (\mtt{x} \mapsto \mtt{empty}) \circ \lambda_k$ (since no repeated list variables)\\ 
&=& $\sigma_{k+1} \circ \lambda_{k+1}$ as required (set $\lambda_{k+1} = \lambda_k$)\\
\end{tabular}
\end{subcase}

\begin{subcase} $Type(\overline{s_m}) = \mtt{empty}$ 

As explained, this subcase cannot occur if $\delta$ is a unifier.

\end{subcase}

\end{case}

\begin{case}  $L(\overline{s},m) > L(\overline{t},n)$

As in the previous case, it follows that $\overline{s}_m$ is a list variable with $|\delta(x)|>0$ for its parent symbol~$\mtt{x}$. By hypothesis, $\mtt{x}$ is also a list variable.

\begin{subcase} Let $LookAhead(\overline{s}_m) = false$, i.e. nothing follows $\mtt{x}$

The unification problem at node $b_k$ must be of the form $\mtt{x} = L$ so rule Subst1 is applicable (condition \ref{en:i} satisfied). Since the resulting unification problem is $\langle \varnothing, \sigma_k \circ \{\mtt{x} \mapsto y:M \} \rangle$ and \textsf{Select} moves the head position to the end of the strings \={s} and \={t}, condition \ref{en:ii} is satisfied. 

For \ref{en:iii}, let $\lambda_k' = (L \mapsto T)$ where $T = \delta(L)$. For $\delta$ to be a unifier, it must contain term(s) $(\mtt{x} \mapsto T)$

\begin{tabular}{lcl}
$\delta$ &$\supseteq$& $\sigma_k \circ \lambda_k \cup (x \mapsto T)$ \\
&=& $\sigma_k \circ \lambda_k \cup ((x \mapsto L) \circ \lambda_k')$ \\
&=& $\sigma_k \circ (x \mapsto L) \circ \lambda_k \circ \lambda_k'$ (since no repeated list variables)\\ 
&=& $\sigma_{k+1} \circ \lambda_{k+1}$ as required (set $\lambda_{k+1} = \lambda_k \circ \lambda_k'$)\\
\end{tabular}

\end{subcase}

\noindent
Now let $LookAhead(\overline{s}_m) = true$ and the problem is of the form $\mtt{x}:L=s:M$ with $L\neq empty$.

\begin{subcase} $Type(\overline{t_n}) = ListVar$
 
By hypothesis, $\mtt{y}$ must be a list variable. Then rule \textsf{Decomp2} is applicable (condition \ref{en:i} satisfied). Because $\mtt{x}$ has the longer expansion by $\delta$, $\overline{s_{m+L(\overline{t}_n)+1}}$ is a child variable of $\mtt{x}$ and condition \ref{en:ii} holds. 

For \ref{en:iii}, let $\lambda_k' = (\overline{s_n} \ldots \overline{s_{m+L(\overline{t}_n)}} \mapsto T)$ where $T = \delta(\overline{s_m} \ldots \overline{s_{m+L(\overline{t}_n)}})$, and the argument becomes identical to Case 1.1 ($\mtt{x}$ starts with $\mtt{y}$, so $\sigma_{k+1} = \sigma_k \circ \{\mtt{x} \mapsto \mtt{y:x'} \}$)

\end{subcase}

\begin{subcase} $Type(\overline{t_n}) = AtomVar$ or $Type(\overline{t_n}) = AtomExpr$

By hypothesis, $\mtt{y}$ is an atom variable or atom expression. Then the problem is of the form $\mtt{x}:L = \mtt{a}:M$ with $L \neq \mtt{empty}$. So rule \textsf{Subst3} is applicable (condition \ref{en:i} satisfied). Positions are incremented by 1, so condition \ref{en:ii} holds as $\overline{s}_{m+1}$ and $\mtt{x'}$ (in \textsf{Orient3}) must have the same parent symbol. 

For \ref{en:iii}, let $\lambda_k' = (\overline{t}_n \mapsto T) $ where $T = \delta(\overline{t}_n)$. It follows that $\delta$ must contain term ($\overline{s}_m \mapsto T$) to be a unifier.

\begin{tabular}{lcl}
$\delta$ &$\supseteq$& $\sigma_k \circ \lambda_k \cup (\overline{s}_m \mapsto T)$ \\
&=& $\sigma_k \circ \lambda_k \cup ((\overline{s}_m \mapsto \overline{t}_n) \circ \lambda_k')$ \\
&=& $\sigma_k \circ \lambda_k \cup ((\mtt{x} \mapsto \mtt{a:x'}) \circ \lambda_k')$ \\
&=& $\sigma_k \circ (\mtt{x} \mapsto \mtt{a:x'}) \circ  \lambda_k \circ \lambda_k'$ (since no repeated list variables) \\
&=& $\sigma_{k+1} \circ \lambda_{k+1}$ as required (set $\lambda_{k+1} = \lambda_k \circ \lambda_k'$)
\end{tabular}
\end{subcase}

\begin{subcase} $Type(\overline{t_n}) = EmptyListVar$

We have that $\overline{t}_n$ is of the form $\mtt{y}_e$ where $\mtt{y}$ is a list variable with $\delta(\mtt{y})=\mtt{empty}$. The unification problem is then of the form $\mtt{x:L=y:M}$. Rule \textsf{Decomp2} becomes applicable

Let $\lambda_k'= (\mtt{y \mapsto empty})$ 

\begin{tabular}{lcl}
$\delta$ &$\supseteq$& $\sigma_k \circ \lambda_k \cup (\mtt{y} \mapsto \mtt{empty})$ \\
&=& $\sigma_k \circ \lambda_k \cup ((\mtt{x} \mapsto \mtt{y:x'}) \circ \lambda_k')$ \\
&=& $\sigma_k \circ (\mtt{x} \mapsto \mtt{y:x'}) \circ \lambda_k \circ \lambda_k'$ (because no repeated list variables) \\
&=& $\sigma_{k+1} \circ \lambda_{k+1}$ as required (set $\lambda_{k+1} = \lambda_k \circ \lambda_k'$) \\
\end{tabular}
\end{subcase}

\end{case}

\begin{case} $L(\overline{s},m) = L(\overline{t},n)$

\begin{subcase} $L(\overline{s},m) = L(\overline{t},n) > 0$

It must be the case that $\overline{s}_m$ and $\overline{t}_n$ are list variables because $L(\overline{s}_m) > 0$ and $L(\overline{t},n) > 0$ only for list variables.

By hypothesis, it follows that $\mtt{x}$ and $\mtt{y}$ are also list variables. 

If there is nothing after the $\mtt{x}$ (i.e. $LookAhead(\overline{s}_m) = false$). Then \textsf{Subst1} is applicable and the argument is the same as Case 2.1.

Otherwise, it must be that $L \neq \mtt{empty}$ so rule \textsf{Decomp1} is applicable. Here the problem is of the form $\mtt{x:L=y:M}$ and the generated substitution is $\sigma_{k+1} = \sigma_k \circ (\mtt{x \mapsto y})$. 

For \ref{en:ii}, the argument is similar to Case 2.2 except that now both $\overline{s}_{m+L(\overline{s_m})+1}$ and $L$, and $\overline{t}_{n+L(\overline{s_m})+1}$ and $M$ must have the same pairs of parent symbols.

Let $\lambda_k' = (\mtt{y \mapsto \delta(y)})$. It must be the case that $\delta$ also contains term $(\mtt{x \mapsto \delta(y)})$ to be a unifier.

\begin{tabular}{lcl}
$\delta$ &$\supseteq$& $\sigma_k \circ \lambda_k \cup (\mtt{x} \mapsto \mtt{\delta(y)})$ \\
&=& $\sigma_k \circ \lambda_k \cup ((\mtt{x} \mapsto \mtt{y}) \circ \lambda_k')$ \\
&=& $\sigma_k \circ (\mtt{x} \mapsto \mtt{y}) \circ \lambda_k \circ  \lambda_k'$ (because no repeated list variables) \\
&=& $\sigma_{k+1} \circ \lambda_{k+1}$ as required (set $\lambda_{k+1} = \lambda_k \circ \lambda_k'$) \\
\end{tabular}

\end{subcase}

\begin{subcase} $L(\overline{s},m) = L(\overline{t},n) = 0$

This is the largest subcase because $\overline{s_m}$ and $\overline{t_n}$ can be of any of the possible types.

\textbf{3.2.1-3 } $(ListVar,ListVar)$ or $(ListVar,AtomVar)$ or $(ListVar,AtomExpr)$-- 

Same as \textbf{3.1} because the parent symbol of $\overline{t}_n$ can be matched by $s$ in \textsf{Decomp1}.

\textbf{3.2.4. } $(ListVar,EmptyListVar)$ -- same as Case \textbf{2.4}

\textbf{3.2.5. } $(ListVar,empty)$ -- cannot happen as $\delta$ is a unifier
\\

\textbf{3.2.6. } $(AtomVar,ListVar)$ -- same as Case \textbf{1.2}

\textbf{3.2.7. } $(AtomVar,AtomVar)$ or $(AtomVar,AtomExpr)$
\begin{itemize}
\item $LookAhead(\overline{s}_m) = true$, i.e. $L \neq \mtt{empty}$ - then the problem is of the form $\mtt{a}:L=\mtt{b}:M$ ($\mtt{a}$ and $\mtt{b}$ are also atom variables by hypothesis) and rule \textsf{Decomp1'} is applicable. \textsf{Unify} generates $\sigma_{k+1} = \sigma_k \circ \{\mtt{a \mapsto b}\}$. 

Let $\lambda_k' = (\overline{s}_m \mapsto \delta(\overline{s}_m)) = (\mtt{a \mapsto \delta(a)})$. Then $\delta$ must also contain term $(\overline{t}_n \mapsto \delta(\overline{s}_m))$ to be a unifier:

\begin{tabular}{lcl}
$\delta$ &$\supseteq$& $\sigma_k \circ \lambda_k \cup (\overline{t}_n \mapsto \delta(\overline{s}_m))$ \\
&=& $\sigma_k \circ \lambda_k \cup ((\overline{t}_n \mapsto \overline{s}_m) \circ \lambda_k')$ \\
&=& $\sigma_k \circ \lambda_k \cup ((\mtt{a \mapsto b}) \circ \lambda_k')$ \\
&=& $\sigma_k \circ (\mtt{a} \mapsto \mtt{b}) \circ \lambda_k \circ  \lambda_k'$ \\
&=& $\sigma_{k+1} \circ \lambda_{k+1}$ as required (set $\lambda_{k+1} = \lambda_k \circ \lambda_k'$) \\
\end{tabular}

\item $LookAhead(\overline{s}_m) = false$, i.e. the problem is $\mtt{a=b:M}$. 

\begin{itemize}
\item If $LookAhead(\overline{t}_n) = true$, then $M$ can only be a single list variable for $\delta$ to be a unifier with $\delta(\mtt{y}) = \mtt{empty}$. This makes rule \textsf{Decomp4} applicable. For \ref{en:iii}, the argument is the same as in the previous item (with $\sigma_{k+1} = \sigma_k \circ \{\mtt{a \mapsto b} \}$).
\item If $LookAhead(\overline{t}_n) = false$, then the problem is of the form $\mtt{a = b}$, which makes \textsf{Subst1} is applicable and the argument is the same as Case 2.1
\end{itemize}
\end{itemize} 

\textbf{3.2.9. } $(AtomVar,EmptyListVar)$ -- same as Case \textbf{1.2}

\textbf{3.2.10. } $(AtomVar,empty)$ -- cannot happen due to failure lemmata  

\textbf{3.2.11. } $(AtomExpr,ListVar)$ -- same as Case \textbf{1.3}
 
\textbf{3.2.12. } $(AtomExpr,AtomVar)$ -- same as Case \textbf{1.3} 

\textbf{3.2.13. } $(AtomExpr,AtomExpr)$ -- 
Both symbols must be the same atom expression due to not considering the subunification algorithm for String-Char. So rule \textsf{Decomp3} is applicable. Conditions \ref{en:ii} and \ref{en:iii} hold trivially.
 
\textbf{3.2.14. } $(AtomExpr,EmptyListVar)$ -- same as \textbf{1.3}
 
\textbf{3.2.16-18. } $(EmptyListVar,ListVar)$ or $(EmptyListVar,AtomVar)$ or $(EmptyListVar,AtomExpr)$ 
It must be the case that there is something following $\overline{s_m}$ for $\delta$ to be a unifier.
Proceed proof like Case \textbf{1.4}

\textbf{3.2.19. } $(EmptyListVar,EmptyListVar)$ -- same as \textbf{3.1}

\textbf{3.2.20. } $(EmptyListVar,empty)$

It follows that $\mtt{x}$ (the parent of $\overline{s}_m$) is a list variable with $\delta(\mtt{x}) = \mtt{empty}$. Also, since there are no repeated list variables, for $\delta$ to be a unifier there must be nothing following $\mtt{x}$. So the problem is of the form $\mtt{x=empty}$ and rule \textsf{Subst1} is applicable.

For condition \ref{en:iii}, proceed as Case \textbf{1.4}
 
\textbf{3.2.21. } $(empty,ListVar)$ -- cannot happen because $\delta$ is a unifier
 
\textbf{3.2.22 \& 3.2.23 } $(empty,AtomVar)$ or $(empty,AtomExpr)$ -- cannot happen due to failure lemmata 

\textbf{3.2.24. } $(empty,EmptyListVar)$

It follows that $\mtt{y}$ (the parent of $\overline{t}_n$) is a list variable with $\delta(\mtt{y}) = \mtt{empty}$. So the problem has the form $\mtt{empty = x}:L$ and rule \textsf{Orient4} is applicable.
 
\textbf{3.2.25. } $(empty,empty)$

It follows that rule \textsf{Remove} is applicable. Conditions \ref{en:ii} and \ref{en:iii} hold trivially.

\end{subcase}

\end{case}
\end{proof}

%% file: conclusion.tex
\section{Conclusion}
\label{sec:conclusion}

This paper presents groundwork for a static confluence analysis of GP programs. We have constructed a rule-based unification algorithm for systems of equations with left-hand expressions of rule schemata, and have shown that the algorithm always terminates and is sound. Moreover, we have proved completeness in that every unifier of the input problem is an instance of some unifier in the computed set of solutions.

Future work includes establishing a Critical Pair Lemma in the sense of \cite{Plump05a}; this entails developing a notion of \emph{independent}\/ rule schema applications, as well as restriction and embedding theorems for derivations with rule schemata. Another topic is to consider critical pairs of conditional rule schemata (see \cite{ehrig2012}).

In addition, since critical pairs contain graphs labelled with expressions, checking joinability of critical pairs will require sufficient conditions under which equivalence of expressions can be decided. This is because the theory of GP's label algebra includes the undecidable theory of Peano arithmetic.

\begin{acknowledge}
We thank the anonymous referees of GCM 2014 for their comments on previous versions of this paper. 
\end{acknowledge}

%% file: ms.bbl
\newcommand{\etalchar}[1]{$^{#1}$}
\begin{thebibliography}{EGH{\etalchar{+}}12}

\bibitem[BFPR15]{Bak-Faulkner-Plump-Runciman15a}
C.~Bak, G.~Faulkner, D.~Plump, C.~Runciman.
A Reference Interpreter for the Graph Programming Language {GP 2}.
In \emph{Proc.\ Graphs as Models (GaM 2015)}.
Electronic Proceedings in Theoretical Computer Science~181, pp.~48--64.
2015.
\\\doi{10.4204/EPTCS.181}

\bibitem[BS01]{Baader-Snyder01a}
F.~Baader, W.~Snyder.
Unification Theory.
In Robinson and Voronkov (eds.), \emph{Handbook of Automated Reasoning}.
Pp.~445--532.
Elsevier and MIT Press, 2001.

\bibitem[EEPT06]{Ehrig-Ehrig-Prange-Taentzer06a}
H.~Ehrig, K.~Ehrig, U.~Prange, G.~Taentzer.
\emph{Fundamentals of Algebraic Graph Transformation}.
Monographs in Theoretical Computer Science.
Springer-Verlag, 2006.

\bibitem[EGH{\etalchar{+}}12]{ehrig2012}
H.~Ehrig, U.~Golas, A.~Habel, L.~Lambers, F.~Orejas.
M-Adhesive Transformation Systems with Nested Application Conditions. Part 2:
  Embedding, Critical Pairs and Local Confluence.
\emph{Fundamenta Informaticae} 118(1):35--63, 2012.

\bibitem[GLEO12]{Golas-Lambers-Ehrig-Orejas12a}
U.~Golas, L.~Lambers, H.~Ehrig, F.~Orejas.
Attributed Graph Transformation with Inheritance: Efficient Conflict Detection
  and Local Confluence Analysis Using Abstract Critical Pairs.
\emph{Theoretical Computer Science} 424:46--68, 2012.
\\\doi{10.1016/j.tcs.2012.01.032}

\bibitem[HKT02]{Heckel-Kuester-Taentzer02a}
R.~Heckel, J.~M. K\"uster, G.~Taentzer.
Confluence of Typed Attributed Graph Transformation Systems.
In \emph{Proc.\ International Conference on Graph Transformation (ICGT 2002)}.
Lecture Notes in Computer Science~2505, pp.~161--176.
Springer-Verlag, 2002.

\bibitem[HP15]{Hristakiev-Plump15a}
I.~Hristakiev, D.~Plump.
A Unification Algorithm for {GP 2}.
In \emph{Graph Computation Models (GCM 2014), Revised Selected Papers}.
Electronic Communications of the EASST~71.
2015.

\bibitem[Jaf90]{Jaffar90a}
J.~Jaffar.
Minimal and Complete Word Unification.
\emph{Journal of the ACM} 37(1):47--85, 1990.

\bibitem[Pla99]{Plandowski99a}
W.~Plandowski.
Satisfiability of Word Equations with Constants is in {PSPACE}.
In \emph{Symposium on Foundations of Computer Science (FOCS 1999)}.
Pp.~495--500.
IEEE Computer Society, 1999.

\bibitem[Plo72]{Plotkin72}
G.~Plotkin.
Building-in Equational Theories.
\emph{Machine intelligence} 7(4):73--90, 1972.

\bibitem[Plu05]{Plump05a}
D.~Plump.
Confluence of Graph Transformation Revisited.
In Middeldorp et~al. (eds.), \emph{Processes, Terms and Cycles: Steps on the
  Road to Infinity: Essays Dedicated to {Jan Willem Klop} on the Occasion of
  His 60th Birthday}.
Lecture Notes in Computer Science~3838, pp.~280--308.
Springer-Verlag, 2005.
\\\doi{10.1007/11601548}

\bibitem[Plu12]{Plump12a}
D.~Plump.
The Design of {GP 2}.
In \emph{Proc.\ International Workshop on Reduction Strategies in Rewriting and
  Programming (WRS 2011)}.
Electronic Proceedings in Theoretical Computer Science~82, pp.~1--16.
2012.
\\\doi{10.4204/EPTCS.82.1}

\bibitem[Sch92]{Schulz92a}
K.~U. Schulz.
Makanin's Algorithm for Word Equations: {T}wo Improvements and a
  Generalization.
In \emph{Proc.\ Word Equations and Related Topics ({IWWERT} '90)}.
Lecture Notes in Computer Science~572, pp.~85--150.
Springer-Verlag, 1992.

\bibitem[Sie78]{Siekmann78a}
J.~Siekmann.
\emph{Unification and Matching Problems}.
PhD thesis, University of Essex, 1978.

\end{thebibliography}
